\documentclass[sigconf, screen, nonacm]{acmart}

\renewcommand\footnotetextcopyrightpermission[1]{}

\usepackage{epsfig,endnotes}
\usepackage{varwidth}

\usepackage{algorithm}
\usepackage{algpseudocode}

\usepackage{xcolor}
\usepackage{nicefrac}
\usepackage{amsmath}
\usepackage{braket}
\usepackage{bm}
\usepackage{mathtools}
\usepackage{multirow}
\usepackage{bigdelim}
\usepackage{mathtools}
\usepackage{amsfonts}

\usepackage{indentfirst}
\usepackage{booktabs}

\usepackage{enumitem}
\usepackage{boxedminipage}
\usepackage{float}
\usepackage{graphicx}
\usepackage[belowskip=0pt,aboveskip=1pt]{caption}
\usepackage{subcaption}
\usepackage[normalem]{ulem}
\usepackage{xpatch}
\usepackage{xspace}
\usepackage{listings}
\usepackage{mathpartir}
\usepackage{makecell}
\usepackage{hyperref}
\usepackage{cleveref}
\usepackage{textcomp}
\usepackage{framed,mdframed}
\usepackage{tablefootnote}

\usepackage{verbatim}
\usepackage{fancyvrb}
\usepackage{comment}
\usepackage{mathtools}

\usepackage{stmaryrd}

\newlength{\saveparindent}
\newlength{\saveparskip}
\newenvironment{tiret}{\begin{list}{\hspace{2pt}\rule[0.5ex]{6pt}{1pt}\hfill}{\labelwidth=15pt\labelsep=5pt \leftmargin=20pt \topsep=3pt\setlength{\listparindent}{\saveparindent}\setlength{\parsep}{\saveparskip}\setlength{\itemsep}{0pt} }}{\end{list}}

\newcommand{\tool}{{\textsc{CrypTFlow2}}\xspace}
\newcommand{\toollib}{{\textsc{SCI}}\xspace}
\newcommand{\toolot}{\textsc{SCI}_{\mathsf{OT}}}
\newcommand{\toolhe}{\textsc{SCI}_{\mathsf{HE}}}

\newcommand{\p}[1]{^{(#1)}}

\newcommand{\namedref}[2]{\hyperref[#2]{#1~\ref*{#2}}\xspace}

\newcommand{\theoremref}[1]{\namedref{Theorem}{theorem:#1}}

\newcommand{\corolref}[1]{\namedref{Corollary}{corol:#1}}
\newcommand{\figureref}[1]{\namedref{Figure}{fig:#1}}
\newcommand{\tableref}[1]{\namedref{Table}{tab:#1}}
\newcommand{\equationref}[1]{\namedref{Equation}{eq:#1}}

\newcommand{\sectionref}[1]{\namedref{Section}{sec:#1}}
\newcommand{\appendixref}[1]{\namedref{Appendix}{app:#1}}

\newcommand{\algoref}[1]{\namedref{Algorithm}{algo:#1}}

\definecolor{forestgreen}{rgb}{0.13, 0.55, 0.13}
\definecolor{cadmiumgreen}{rgb}{0.0, 0.42, 0.24}

\definecolor{mypink}{rgb}{1,0.2,0.4}

\newcommand{\adv}{\mathcal{A}}
\newcommand{\env}{\mathcal{Z}}
\newcommand{\prot}{\Pi}

\newcommand{\simu}{\mathcal{S}}
\newcommand{\F}{\mathcal{F}}
\newcommand{\W}{\mathcal{W}}
\newcommand{\secpar}{\lambda}

\newcommand{\zo}{\{0,1\}}

\newcommand{\share}[3]{\langle #1\rangle^{#2}_{#3}}

\newcommand{\genshare}[2]{\mathsf{Share}^{#1}(#2)}
\newcommand{\reconst}[2]{\mathsf{Reconst}^{#1}(#2)}

\newcommand{\bbZ}{\mathbb{Z}}

\newcommand{\relu}{\mathsf{ReLU}}
\newcommand{\corr}{\mathsf{corr}}

\newcommand{\drelu}{\mathsf{DReLU}}
\newcommand{\gt}{\mathsf{gt}}
\newcommand{\maxpool}{\mathsf{Maxpool}}

\newcommand{\resnet}{\textsf{ResNet50}\xspace}
\newcommand{\squeezenet}{\textsf{SqueezeNet}\xspace}
\newcommand{\densenet}{\textsf{DenseNet121}\xspace}

\newcommand{\avgpool}{\mathsf{Avgpool}}

\newcommand{\argmax}{\mathsf{Argmax}}

\newcommand{\party}[1]{P_{#1}}
\newcommand{\func}{\mathcal{F}}

\newcommand{\xor}{\oplus}
\newcommand{\fidiv}[2]{\mathsf{idiv}(#1, #2)}
\newcommand{\frdiv}[2]{\mathsf{rdiv}(#1, #2)}
\newcommand{\idiv}{\mathsf{idiv}}
\newcommand{\rdiv}{\mathsf{rdiv}}

\newcommand{\mill}{\mathsf{MILL}}
\newcommand{\protmill}[1]{\prot_{\mill}^{#1}}
\newcommand{\fmill}[1]{\func_{\mathsf{MILL}}^{#1}}
\newcommand{\fand}{\func_{\mathsf{AND}}}
\newcommand{\protreluint}[1]{\prot_{\mathsf{ReLU}}^{\mathsf{int},#1}}
\newcommand{\protdreluint}[1]{\prot_{\mathsf{DReLU}}^{\mathsf{int},#1}}
\newcommand{\protreluring}[1]{\prot_{\mathsf{ReLU}}^{\mathsf{ring},#1}}
\newcommand{\fdreluint}[1]{\func_{\mathsf{DReLU}}^{\mathsf{int},#1}}

\newcommand{\protdreluring}[1]{\prot_{\mathsf{DReLU}}^{\mathsf{ring},#1}}
\newcommand{\fdreluring}[1]{\func_{\mathsf{DReLU}}^{\mathsf{ring},#1}}

\newcommand{\protdreluringsimple}[1]{\prot_{\mathsf{DReLU^{simple}}}^{\mathsf{ring},#1}}

\newcommand{\fmux}[1]{\func_{\mathsf{MUX}}^{#1}}
\newcommand{\protmux}[1]{\prot_{\mathsf{MUX}}^{#1}}

\newcommand{\fBtoA}[1]{\func_{\mathsf{B2A}}^{#1}}
\newcommand{\protbtoa}[1]{\prot_{\mathsf{B2A}}^{#1}}

\newcommand{\prottruncint}[2]{\prot_{\mathsf{Trunc}}^{\mathsf{int},#1,#2}}

\newcommand{\protdivint}[1]{\prot_{\mathsf{DIV}}^{\mathsf{int},#1}}
\newcommand{\protdivring}[1]{\prot_{\mathsf{DIV}}^{\mathsf{ring},#1}}
\newcommand{\ftruncint}[2]{\func_{\mathsf{Trunc}}^{\mathsf{int},#1,#2}}
\newcommand{\fdivring}[2]{\func_{\mathsf{DIV}}^{\mathsf{ring},#1,#2}}

\newcommand{\ringsz}{n}
\newcommand{\polymod}{N}

\newcommand{\blsize}{m}
\newcommand{\blnum}{q}
\newcommand{\alice}{\ensuremath{P_0}\xspace}
\newcommand{\bob}{\ensuremath{P_1}\xspace}
\newcommand{\ind}[1]{\ensuremath{\mathbf{1}\{ #1 \}}\xspace}
\newcommand{\concat}{||}
\newcommand{\getsr}[0]{\mathbin{\stackrel{\mbox{\,\tiny \$}}{\gets}}}
\newcommand{\twoblsize}{M}
\newcommand{\lessthan}{\mathsf{lt}}
\newcommand{\equal}{\mathsf{eq}}

\newcommand{\temp}{\mathsf{temp}}

\newcommand{\intx}{L}
\newcommand{\msbs}{\mathsf{msb}}
\newcommand{\msbl}{\mathsf{MSB}}
\newcommand{\carry}{\mathsf{carry}}

\newcommand{\etax}{\eta}
\newcommand{\wrap}{\mathsf{wrap}}
\newcommand{\lt}{\mathsf{lt}}
\newcommand{\rt}{\mathsf{rt}}
\newcommand{\xt}{\mathsf{xt}}

\newcommand{\kkot}[2]{\ensuremath{{#1 \choose 1}\text{-}\mathsf{OT}_{#2}}\xspace}
\newcommand{\iknpcot}[1]{\ensuremath{{2 \choose 1}\text{-}\mathsf{COT}_{#1}}\xspace}

\makeatletter
\renewcommand\theHALG@line{\thealgorithm.\arabic{ALG@line}}
\makeatother

\begin{document}
\pagestyle{plain}
\title{\tool: Practical 2-Party Secure Inference} 

\author[D. Rathee]{Deevashwer Rathee}
\affiliation{\institution{Microsoft Research}
}
\email{t-dee@microsoft.com}

\author[M. Rathee]{Mayank Rathee}
\affiliation{\institution{Microsoft Research}
}
\email{t-may@microsoft.com}

\author[N. Kumar]{Nishant Kumar}
\affiliation{\institution{Microsoft Research}
}
\email{nishant.kr10@gmail.com}

\author[N. Chandran]{Nishanth Chandran}
\affiliation{\institution{Microsoft Research}
}
\email{nichandr@microsoft.com}

\author[D. Gupta]{Divya Gupta}
\affiliation{\institution{Microsoft Research}
}
\email{divya.gupta@microsoft.com}

\author[A. Rastogi]{Aseem Rastogi}
\affiliation{\institution{Microsoft Research}
}
\email{aseemr@microsoft.com}

\author[R. Sharma]{Rahul Sharma}
\affiliation{\institution{Microsoft Research}
}
\email{rahsha@microsoft.com}

\begin{abstract}

We present $\tool$, a cryptographic framework
for secure inference over realistic Deep Neural Networks (DNNs) using
secure 2-party computation. 
$\tool$ protocols are both correct -- i.e., their outputs are
bitwise equivalent to the cleartext execution -- and
efficient -- they outperform the state-of-the-art protocols in both
latency and scale. At the core of \tool, we have new
2PC protocols for secure comparison
and division, designed carefully to balance round and
communication complexity for secure inference tasks. Using \tool, we present
the first secure
inference over ImageNet-scale DNNs like $\resnet$ and $\densenet$. These DNNs are at least an
order of magnitude larger than those considered in the prior
work of 2-party DNN inference. 
Even on the benchmarks considered by prior work, \tool  requires
an order of magnitude less communication and $20\times$-$30\times$ less time than 
the state-of-the-art.
\end{abstract}


\maketitle

\section{Introduction}
\label{sec:intro}
The problem of privacy preserving machine learning has become increasingly important. Recently, there have been many works that have made rapid strides towards realizing {\em secure inference}~\cite{cryptonets,chet,secureml,minionn,gazelle,delphi,nitin,xonn,ezpc,hycc,nhe1,deepsecure,ball}.
Consider a server that holds the weights $w$ of a publicly known deep neural network (DNN), $F$, that has been trained on private data. A client holds a private input $x$; in a standard machine learning (ML) inference task, the goal is for the client to learn the prediction $F(x,w)$ of the server's model on the input $x$. In secure inference, the inference is performed with the guarantee that the server learns {\em nothing} about $x$ and the client learns nothing about the server's model $w$ beyond what can be deduced from $F(x,w)$ and $x$. 

A solution for secure inference that scales to practical ML tasks would open a plethora of applications based on MLaaS (ML as a Service). Users can obtain value from ML services  without worrying about the loss of their private data, while model owners can effectively monetize their services with no fear of breaches of client data (they never observe private client data in the clear). Perhaps the most important emerging applications for secure inference  are in healthcare where
prior work~\cite{cryptflow,xonn,nitin} has explored secure inference services for privacy preserving medical diagnosis of chest diseases, diabetic retinopathy, malaria, and so on.

Secure inference is an instance of secure 2-party computation (2PC) and cryptographically secure general protocols for 2PC have been known for decades~\cite{Yao,gmw}. 
However, secure inference for practical ML tasks, e.g., ImageNet scale prediction~\cite{imagenet}, is challenging for two reasons: 
a) realistic DNNs use $\relu$ activations\footnote{$\relu(x)$ is defined as $\max(x,0)$.} that are expensive to compute securely; and b) preserving inference accuracy requires a faithful implementation of secure fixed-point arithmetic.
 All prior works~\cite{cryptonets,secureml,minionn,gazelle,delphi,ball} fail to provide efficient implementation of $\relu$s. Although $\relu$s can be replaced with approximations that are more tractable for 2PC~\cite{cryptonets,delphi,chet}, this approach results in significant accuracy losses that can degrade user experience. The only known approaches to evaluate $\relu$s efficiently require sacrificing security by making the untenable assumption that a non-colluding third party takes part in the protocol~\cite{chameleon,securenn,cryptflow,aby3,quantizednn} or by leaking activations~\cite{nhe2}.
Moreover, some prior works~\cite{secureml,delphi,securenn,aby3,cryptflow} even sacrifice correctness of their fixed-point implementations and the result of their secure execution can sometimes diverge from the expected result, i.e. cleartext execution, in random and unpredictable ways.
Thus, correct and efficient 2PC protocols for secure inference over realistic DNNs remain elusive.
\subsection{Our Contributions}

In this work, we address the above two challenges and build new semi-honest secure 2-party cryptographic protocols for secure computation of DNN inference. 
Our new efficient protocols enable the first secure implementations of ImageNet scale inference that complete in under a minute!
We make three main  contributions:
\begin{tiret}
\item First, we give new protocols for millionaires' and $\drelu$\footnote{$\drelu$ is the derivative of $\relu$, i.e., $\drelu(x)$ is $1$ if $x \geq 0$ and $0$ otherwise.} that enable us to securely and efficiently evaluate the non-linear layers of DNNs such as $\relu$, $\maxpool$ and $\argmax$. \item Second, we provide new protocols for division. Together with new theorems that we prove on fixed-point arithmetic over shares, we show how to evaluate linear layers, such as convolutions, average pool and fully connected layers, faithfully. 
\item Finally, by providing protocols that can work on a variety of input domains, we build a system\footnote{Implementation is available at \url{https://github.com/mpc-msri/EzPC}.}
    $\tool$ that supports two different types of Secure and Correct Inference (\toollib) protocols where linear layers can be evaluated using either homomorphic encryption ($\toolhe$) or through oblivious transfer ($\toolot$).
\end{tiret}
We now provide more details of our main contributions. 
\\\\
\noindent{\textbf{New millionaires' and $\drelu$ protocols.}} 
Our first main technical contribution is a novel protocol for the well-known {\em millionaires'} problem~\cite{Yao}, where parties \alice and \bob hold  $\ell-$bit integers $x$ and $y$, respectively, and want to securely compute $x < y$ (or, secret shares of $x<y$).
The theoretical communication complexity of our protocol is $\approx 3\times$ better than the most communication efficient prior millionaires' protocol~\cite{Cou18,Yao,gmw,emp-toolkit,GSV07}.
In terms of round complexity, our protocol executes in $\log \ell$ rounds (e.g. $5$ rounds for $\ell=32$ bits);
see \tableref{comp-mill} for a detailed comparison and \cite{Cou18} for a detailed overview of the costs of other comparison protocols.

Using our protocol for millionaires' problem, we build new and efficient protocols for computing $\drelu$ for both $\ell-$bit integers (i.e., $\bbZ_\intx$, $\intx = {2^\ell}$) and general rings $\bbZ_\ringsz$.
Our protocol for $\drelu$ serves as one of the main building blocks for non-linear activations such as $\relu$ and $\maxpool$, as well as division over both input domains.
Providing support for  $\ell-$bit integers $\bbZ_{\intx}$ as well as arbitrary rings $\bbZ_{\ringsz}$, allows us to securely evaluate the linear layers (such as matrix multiplication and convolutions) using  the approaches of Oblivious Transfer (OT)~\cite{beaver,secureml} as well as Homomorphic Encryption (HE)~\cite{gentryfhe,gazelle,delphi}, respectively. 
This provides our protocols great flexibility when executing over different network configurations. 
Since all prior work~\cite{secureml,minionn,gazelle,delphi} for securely computing these activations rely on Yao's garbled circuits~\cite{Yao}, our protocols are much more efficient in both settings.
Asymptotically, our $\relu$ protocol over $\bbZ_\intx$ and $\bbZ_\ringsz$ communicate $\approx 8\times$ and $\approx 12\times$ less bits than prior works~\cite{Yao,secureml,minionn,gazelle,delphi,emp-toolkit} (see \tableref{comp-relu} for a detailed comparison).
Experimentally, our protocols are at least an order of magnitude more performant than prior protocols when computing $\relu$ activations at the scale of ML applications.
\\\\
\noindent{\textbf{Fixed-point arithmetic.}} 
The ML models used by all prior works on secure inference are expressed
using fixed-point arithmetic; such models can be obtained from~\cite{tflite,bnn,cryptflow,dfq}. A faithful implementation of fixed-point arithmetic is quintessential to ensure that the secure computation is {\em correct}, i.e., it is {\it equivalent} to the cleartext computation for all possible inputs.
Given a secure inference task $F(x,w)$, some prior works~\cite{secureml,delphi,securenn,aby3,cryptflow} give up on correctness
when implementing division operations and instead compute an approximation $F'(x,w)$. In fixed-point arithmetic, each multiplication requires a division by a power-of-2 and multiplications are used pervasively in linear-layers of DNNs. 
 Moreover, layers like average-pool require division for computing means.
Loss in correctness is worrisome as the errors can accumulate and $F'(x,w)$ can be arbitrarily far from $F(x,w)$.
 Recent work~\cite{delphi} has shown that even in practice the approximations can lead  to significant losses in classification accuracy.

As our next contribution, we provide novel protocols to compute division by power-of-2 as well as division by arbitrary integers that are both correct and efficient. The inputs to these protocols can be encoded over both $\ell-$bit integers $\bbZ_{\intx}$ as well as $\bbZ_{\ringsz}$, for arbitrary $\ringsz$. 
To the best of our knowledge, the only known approach to compute division correctly is via garbled circuits which we compare with in Table~\ref{tab:comp-avgpool}.
While garbled circuits based protocols require communication which is quadratic in $\ell$ or $\log n$, our protocols are asymptotically better and incur only linear communication. Concretely, for average pool with $7\times 7$ filters and 32-bit integers, our protocols have $\approx 54\times$ less communication.
\\\\
\noindent{\textbf{Scaling to practical DNNs.}} These efficient protocols, help us securely evaluate practical DNNs like \squeezenet\   on ImageNet scale classification tasks in under a minute. In sharp contrast, all prior works on secure  2-party inference (\cite{cryptonets,chet,secureml,minionn,gazelle,delphi,nitin,xonn,ezpc,hycc,nhe1,deepsecure,ball}) has been limited to small DNNs on tiny datasets like MNIST and CIFAR.
While MNIST deals with the task of classifying black and white handwritten digits given as $28\times 28$ images into the classes 0 to 9, ImageNet tasks are much more complex: typically $224\times 224$ colored images need to be classified into thousand classes (e.g., agaric, gyromitra, ptarmigan, etc.) that even humans can find challenging .
 Additionally, our work is the first to securely evaluate practical {\em convolutional neural networks} (CNNs) like \resnet\ and \densenet; these DNNs are at least an order of magnitude larger than the DNNs considered in prior work, provide over $90\%$ Top-5 accuracy on ImageNet, and have also been shown to predict lung diseases from chest X-ray images~\cite{cryptflow, chestxray2018}.
Thus, our work provides the first implementations of practical ML inference tasks running securely.  Even on the smaller CIFAR scale DNNs, our protocols require an order of magnitude less communication and $20\times$-$30\times$ less time than the state-of-the-art~\cite{delphi} (see  Section~\ref{sec:exp-comparison-prior-work}). 
\\\\
\noindent{\textbf{OT vs HE.}} Through our evaluation, we also resolve the OT vs HE conundrum: 
although the initial works on secure inference~\cite{secureml,minionn} used OT-based protocols for evaluating convolutions, the state-of-the-art protocols~\cite{gazelle,delphi}, which currently provide the best published inference latency, use HE-based convolutions. HE-based secure inference has much less communication than OT but HE requires more computation.
Hence, at the onset of this work, it was not clear to us whether HE-based convolutions would provide us the best latency for ImageNet-scale benchmarks. 

To resolve this empirical question, we implement two classes of protocols, $\toolot$ and $\toolhe$, in \tool.
In $\toolot$, inputs are in $\bbZ_{\intx}$ ($\intx = 2^\ell$, for a suitable choice of $\ell$). Linear layers such as matrix multiplication and convolution are performed using OT-based techniques~\cite{beaver,secureml}, while the activations such as $\relu$, $\maxpool$ and $\avgpool$  are implemented using our new protocols over $\bbZ_{\intx}$. 
In $\toolhe$, inputs are encoded in an appropriate prime field $\bbZ_{\ringsz}$ (similar to~\cite{gazelle,delphi}). 
Here, we compute linear layers using homomorphic encryption and the activations  using our protocols over $\bbZ_{\ringsz}$.
In both $\toolot$ and $\toolhe$ faithful divisions after linear layers are performed using our new protocols over corresponding rings. Next, we evaluate ImageNet-scale inference tasks with both $\toolot$ and $\toolhe$ . We observe that in a WAN setting, where communication is a bottleneck, HE-based inference is always faster and in a LAN setting OT and HE are incomparable. 
\begin{table}
  \centering
      \begin{tabular}{|c|c|c|c|}
    \hline
				Layer & Protocol & Comm. (bits) & Rounds \\ \hline \hline
				\multirowcell{4}{Millionaires' \\ on $\zo^\ell$}
				& GC \cite{Yao, emp-toolkit} & $4\secpar\ell$ & 2 \\ \cline{2-4}
				& GMW$\footnotemark/$GSV~\cite{gmw,GSV07}& $\approx 6\secpar\ell$ & $\log \ell+3$ \\ \cline{2-4}
				& SC3\footnotemark\cite{Cou18} & $> 3\secpar\ell$ & $\approx 4\log^* \secpar$ \\ \cline{2-4}
				& This work ($\blsize = 4$) & $< \secpar \ell + 14\ell $ & $\log {\ell}$\\ \hline \hline
				\multirowcell{5}{Millionaires' \\example \\$\ell = 32$}
				& GC \cite{Yao, emp-toolkit} & 16384 & 2 \\ \cline{2-4}
				& GMW$/$GSV \cite{gmw, GSV07} & 23140 & 8 \\ \cline{2-4}
				& SC3 \cite{Cou18} & 13016 & 15 \\ \cline{2-4}
				& This work ($\blsize = 7$) & 2930 & 5 \\ \cline{2-4}
				& This work ($\blsize = 4$) & 3844 & 5 \\ \hline
			\end{tabular}
 \caption{Comparison of communication with prior work for millionaires' problem. For our protocol, $\blsize$ is a parameter. For concrete bits of communication we use $\secpar=128$. }
\label{tab:comp-mill}
\end{table}
\addtocounter{footnote}{-1}
\footnotetext{Here we state the communication numbers for GMW~\cite{gmw} for a depth-optimized circuit. The circuit that would give the best communication would still have a complexity of $> 2\secpar\ell$ and would additionally pay an inordinate cost in terms of rounds, namely $\ell$.}
\addtocounter{footnote}{1}
\footnotetext{Couteau~\cite{Cou18} presented multiple protocols; we pick the one that has the best communication complexity.}
 \begin{table}
  \centering
      \begin{tabular}{|c|c|c|c|}
    \hline
				Layer & Protocol & Comm. (bits) & Rounds \\ \hline \hline
				\multirowcell{2}{$\relu$ for \\$\bbZ_{2^\ell}$}

				& GC \cite{Yao, emp-toolkit} & $ 8\secpar\ell-4\secpar$ & 2 \\ \cline{2-4}
				& This work & \multirow{1}{*}{ $< \secpar\ell+18\ell$ } & $\log \ell + 2$\\
				 \hline

				\multirowcell{2}{$\relu$ for \\general $\mathbb{Z}_n$}

				& GC \cite{Yao, emp-toolkit} & $18\secpar\etax-6\secpar$ & 2 \\ \cline{2-4}
				& This work & $< \frac{3}{2}\secpar (\etax+1) + 31\etax$ & $\log \etax + 4$\\
				\hline \hline

				\multirowcell{2}{$\relu$ for \\ $\bbZ_{2^\ell}$, $\ell = 32$}
				
				& GC \cite{Yao, emp-toolkit} & 32256 & 2 \\ \cline{2-4}
				& This work  & 3298 & 7 \\ \cline{2-4}
				\hline
				
				\multirowcell{2}{$\relu$ for \\$\mathbb{Z}_n$, $ \etax = 32$}
				
				& GC \cite{Yao, emp-toolkit} & 72960 & 2 \\ \cline{2-4}
				& This work  & 5288 & 9 \\ \cline{2-4}
				\hline
			\end{tabular}
 \caption{Comparison of communication with garbled circuits for $\relu$. We define $\etax = \lceil \log n \rceil$. For concrete bits of communication we use $\secpar=128$. }
\label{tab:comp-relu}
\end{table}
 \begin{table}
  \centering
      \begin{tabular}{|c|c|c|c|}
    \hline
				Layer & Protocol & Comm. (bits) & Rounds \\ \hline \hline
				\multirowcell{2}{$\avgpool_d$ \\$\bbZ_{2^\ell}$}

				& GC \cite{Yao, emp-toolkit} & $2\secpar(\ell^2 + 5\ell -3)$ & 2 \\ \cline{2-4}
				& This work &  $< (\secpar + 21) \cdot ( \ell + 3 \delta )$   & $\log (\ell\delta ) + 4$\\
              \hline

				\multirowcell{2}{$\avgpool_d$  \\ $\mathbb{Z}_n$}

				& GC \cite{Yao, emp-toolkit} & $2\secpar(\etax^2 + 9\etax-3)$ & 2 \\ \cline{2-4}
				& This work & $< (\textstyle\frac{3}{2} \secpar + 34) \cdot (\etax + 2 \delta)$ & $\log (\etax\delta) + 6$\\
             \hline \hline

                \multirowcell{2}{$\avgpool_{49}$ \\  $\bbZ_{2^\ell}, \ell = 32$}
				
				& GC \cite{Yao, emp-toolkit} & 302336 & 2 \\ \cline{2-4}
				& This work  & 5570 & 10 \\ \cline{2-4}
				\hline
				
                \multirowcell{2}{$\avgpool_{49}$  \\ $\mathbb{Z}_n, \etax = 32$}
				
				& GC \cite{Yao, emp-toolkit} & 335104 & 2 \\ \cline{2-4}
				& This work & 7796 & 14 \\ \cline{2-4}
				\hline
			\end{tabular}
 \caption{Comparison of communication with garbled circuits for $\avgpool_d$. We define $\etax = \lceil \log n \rceil$ and $\delta = \lceil \log (6 \cdot d) \rceil$. For concrete bits of communication we use $\secpar=128$. Choice of $d =49$ corresponds to average pool filter of size $7\times 7$. } \label{tab:comp-avgpool}
\end{table}
 \subsection{Our Techniques}

\noindent\textbf{Millionaires'.} Our protocol for securely computing the millionaires' problem (the bit $x<y$) is based on the following observation (first made in~\cite{GSV07}). 
Let $x = x_1||x_0$ and $y = y_1||y_0$ (where $||$ denotes concatenation and $x_1, y_1$ are strings of the same length). 
Then, $x<y$ is the same as checking if either $x_1<y_1$ or $x_1=y_1$ and $x_0<y_0$. 
Now, the original problem is reduced to computing two millionaires' instances over smaller length strings ($x_1<y_1$ and $x_0<y_0$) and one equality test ($x_1 = y_1$). 
By continuing recursively, one could build a tree all the way where the leaves are individual bits, at which point one could use 1-out-of-2 OT-based protocols to perform the comparison/equality. 
However, the communication complexity of this protocol is still quite large. 
We make several important modifications to this approach. 
First, we modify the tree so that the recursion is done $\log (\ell/\blsize)$ times to obtain leaves with strings of size $\blsize$, for a parameter $\blsize$.
We then use 1-out-of-$2^\blsize$ OT to compute the comparison/equality at the leaves, employing the lookup-table based approach of \cite{DKSSZZ17}.
Second, we observe that by carefully setting up the receiver's and sender's messages in the OT protocols for leaf comparisons and equality, multiple 1-out-of-$2^\blsize$ OT instances can be combined to reduce communication. 
Next, recursing up from the leaves to the root, requires securely computing the $\mathsf{AND}$ functionality\footnote{This functionality takes as input shares of bits $x, y$ from the two parties and outputs shares of $x~\mathsf{AND}~y$ to both parties.}  that uses Beaver bit triples~\cite{beaver}. 
We observe that the same secret value is used in 2 $\mathsf{AND}$ instances.  
Hence, we construct correlated pairs of bit triples using 1-out-of-8 OT protocols~\cite{kkot} to reduce this cost to $\secpar+8$ bits (amortized) per triple, where $\secpar$ is the security parameter and typically $128$.
Some more work is needed for the above technique to work efficiently for the general case when $\blsize$ does not divide $\ell$ or $\ell/\blsize$ is not a power of $2$. 
Finally, by picking $\blsize$ appropriately, we obtain a protocol for millionaires' whose concrete communication (in bits) is nearly $5$ times better than prior work. 
\\

\noindent\textbf{$\drelu$.}
Let $a$ be additively secret shared as $a_0, a_1$ over the appropriate ring. 
$\drelu(a)$ is $1$ if $a \geq 0$ and $0$ otherwise; note that $a\geq0$ is defined differently for $\ell-$bit integers and general rings.
Over $\bbZ_\intx$, where values are encoded using 2's complement notation, $\drelu(a) = 1\xor \msbl(a)$, where $\msbl(a)$ is the most significant bit of $a$.
Moreover,  $\msbl(a) = \msbl(a_0) \oplus \msbl(a_1) \oplus \carry$. 
Here, $\carry = 1$ if $a_0'+a_1' \geq 2^{\ell-1}$, where $a'_0, a'_1$ denotes the integer represented by the lower $\ell-1$ bits of $a_0, a_1$.
We compute this $\carry$ bit using a call to our millionaires' protocol.
Over $\bbZ_n$, $\drelu(a) = 1$ if $a \in [0,\lceil n/2 \rceil)$. Given the secret shares $a_0, a_1$, this is equivalent to $(a_0+a_1) \in [0,\lceil n/2 \rceil) \cup [n, \lceil 3n/2 \rceil)$ over integers.
While this can be na\"{\i}vely computed by making 3 calls to the millionaires' protocol, we show that by carefully selecting the inputs to the millionaires' protocol, one can do this with only 2 calls.
Finally, we set things up so that the two calls to millionaires' have correlated inputs that reduces the overall cost to $\approx$ 1.5 instances of millionaires' over $\bbZ_n$.
\\

\noindent\textbf{Division and Truncation.} As a technical result, we  provide a {\em correct} decomposition of division of a secret ring element in $\bbZ_\intx$ or $\bbZ_\ringsz$ by a public integer into division of secret shares by the same public integer and  correction terms (\theoremref{general-division}). 
These correction terms consist of multiple inequalities on secret values. 
As a corollary, we also get a much simpler expression for the special case of {\em truncation}, i.e., dividing $\ell$-bit integers by a power-of-$2$ (\corolref{truncate-int}). 
We believe that the general theorem as well as the corollary can be of independent interest. 
Next, we give efficient protocols for both general division (used for $\avgpool$, \tableref{comp-avgpool}) as well as division by a power-of-$2$ (used for multiplication in fixed-point arithmetic). The inequalities in the correction term are computed using our new protocol for millionaires' and the division of shares can be done locally by the respective parties. 
Our technical theorem is the key to obtaining secure implementation of DNN inference tasks that are bitwise equivalent to cleartext fixed-point execution.

\subsection{Other Related Work}

Perhaps the first work to consider the secure computation of machine learning inference algorithms was that of~\cite{bost}.
SecureML~\cite{secureml} was the first to consider secure neural network inference and training. Apart from the works mentioned earlier, other works include those that considered malicious adversaries~\cite{helen,mlspdz,leviosa} (for simpler ML models like linear models, regression, and polynomials) as well as specialized DNNs with 1 or 2 bit weights~\cite{xonn,deepsecure,nitin}.
Recently, \cite{edabits} gave protocols for faithful truncation (but not division) over $\ell$-bit integers and prime fields in various adversarial settings. For 2-party semi-honest setting, our protocols have up to $20\times$ less communication for the truncations required in our evaluation.
\cite{RSS19} proposed an HE-based triple generation protocol over $\bbZ_{2^{\ell}}$, which requires less communication than generating triples using OT.

\subsection{Organisation}
We begin with the details on security and cryptographic primitives used in Section~\ref{sec:prelims} on preliminaries. 
In Section~\ref{sec:comp-protocols} we provide our protocols for millionaires' (Section~\ref{sec:prot-mill}) and $\drelu$ (Section~\ref{sec:prot-drelu-int},~\ref{sec:prot-drelu-ring}), over both $\bbZ_\intx$ and general ring $\bbZ_\ringsz$. In Section~\ref{sec:division}, we present our protocols for division and truncation. We describe the various components of DNN inference  in Section~\ref{sec:secureinference} and show how to construct secure protocols for all these components given our protocols from Sections~\ref{sec:comp-protocols} and ~\ref{sec:division}. We present our implementation details in Section~\ref{sec:impl} and our experiments in Section~\ref{sec:experiments}. Finally, we conclude and discuss future work in \sectionref{conclusion}.

 \section{Preliminaries}
\label{sec:prelims}
\paragraph{Notation.}
For a set $\W$, $w \getsr \W$ denotes sampling an element $w$, uniformly at random from $\W$. $[\ell]$ denotes the set of integers $\{0, \cdots, \ell-1\}$. Let $\ind{b}$ denote the indicator function that is $1$ when $b$ is {\em true} and $0$ when $b$ is {\em false}.

\subsection{Threat Model and Security}
We provide security in the simulation paradigm~\cite{gmw, canetti00, lindellsim} against a {\em static semi-honest} probabilistic polynomial time (PPT) adversary $\adv$. That is, a computationally bounded adversary $\adv$ corrupts either \alice or \bob at the beginning of the protocol and follows the protocol specification honestly. 
Security is modeled by defining two interactions: a real interaction where \alice and \bob execute the protocol in the presence of $\adv$ and the environment $\env$ and an ideal interaction where the parties send their inputs to a trusted functionality that performs the computation faithfully. 
Security requires that for every adversary $\adv$ in the real interaction, there is an adversary
$\simu$ (called the simulator) in the ideal interaction, such that no environment $\env$ can distinguish between real and ideal interactions. Many of our protocols invoke multiple sub-protocols and we describe these using the \emph{hybrid model}. This is similar to a real interaction, except that sub-protocols are replaced by the invocations of  instances of  corresponding functionalities. A protocol invoking a functionality $\F$ is said to be in ``$\F$-hybrid model.''

\subsection{Cryptographic Primitives}

\subsubsection{Secret Sharing Schemes}
\label{sec:ss}

Throughout this work, we use 2-out-of-2 additive secret sharing schemes over different rings~\cite{shamir,blakeley}. The 3 specific rings that we consider are the field $\bbZ_2$, the ring $\bbZ_\intx$, where $\intx = 2^\ell$ ($\ell = 32$, typically), and the ring $\bbZ_\ringsz$,  for a positive integer $\ringsz$ (this last ring includes the special case of prime fields used in the works of \cite{gazelle,delphi}). We let $\genshare{\intx}{x}$ denote the algorithm that takes as input an element $x$ in $\bbZ_\intx$ and outputs shares over $\bbZ_\intx$, denoted by $\share{x}{\intx}{0}$ and $\share{x}{\intx}{1}$. 
Shares are generated by sampling random ring elements $\share{x}{\intx}{0}$ and $\share{x}{\intx}{1}$, with the only constraint that $\share{x}{\intx}{0}+\share{x}{\intx}{1} = x$ (where $+$ denotes addition in $\bbZ_\intx$). Additive secret sharing schemes are perfectly hiding, i.e., given a share $\share{x}{\intx}{0}$ or $\share{x}{\intx}{1}$, the value $x$ is completely hidden. The reconstruction algorithm $\reconst{\intx}{\share{x}{\intx}{0}, \share{x}{\intx}{1}}$ takes as input the two shares and outputs $x = \share{x}{\intx}{0}+\share{x}{\intx}{1}$. Shares (along with their corresponding $\genshare{}{}$ and $\reconst{}{}$ algorithms) are defined in a similar manner for $\bbZ_2$ and $\bbZ_\ringsz$ with superscripts $B$ and $\ringsz$, respectively. We sometimes refer to shares over $\bbZ_\intx$ and $\bbZ_\ringsz$ as arithmetic shares and shares over $\bbZ_2$ as boolean shares.

\subsubsection{Oblivious Transfer}

Let $\kkot{k}{\ell}$ denote the 1-out-of-$k$ Oblivious Transfer (OT) functionality~\cite{bcr1outofnot} (which generalizes 1-out-of-2 OT~\cite{rabinot, eglot}). The sender's inputs to the functionality are the $k$ strings $m_0, \cdots, m_{k-1}$, each of length $\ell$ and the receiver's input is a value $i \in [k]$. 
The receiver obtains $m_i$ from the functionality and the sender receives no output.
 We use the protocols from~\cite{kkot}, which are an optimized and generalized version of the OT extension framework proposed in~\cite{beaverotextension,iknp}. 
This framework allows the sender and receiver,  to ``reduce''  $\secpar^c$ number of oblivious transfers to $\secpar$ ``base'' OTs. We also use the notion of correlated 1-out-of-2 OT~\cite{ALSZ13}, denoted by $\iknpcot{\ell}$. In our context, this is a functionality where the sender's input is a ring element $x$ and the receiver's input is a choice bit $b$. The sender receives a random ring element $r$ as output and the receiver obtains either $r$ or $x+r$ as output depending on $b$.
The protocols for $\kkot{k}{\ell}$ \cite{kkot} and $\iknpcot{\ell}$ \cite{ALSZ13} execute in $2$ rounds and have total communication\footnote{The protocol of $\kkot{k}{\ell}$ \cite{kkot} incurs a communication cost of $\secpar + k\ell$. However, to achieve the same level of security, their security parameter needs to be twice that of $\iknpcot{\ell}$. In concrete terms, therefore, we write the cost as $2\secpar+k\ell$.} of $2\secpar + k\ell$ and $\secpar + \ell$, respectively. Moreover, simpler $\kkot{2}{\ell}$ has a communication of $\secpar+2\ell$ bits \cite{iknp,ALSZ13}.

\subsubsection{Multiplexer and B2A conversion}
\label{sec:mux}

The functionality $\fmux{\ringsz}$ takes as input arithmetic shares of $a$ over $\ringsz$ and boolean shares of choice bit $c$ from $\party{0}, \party{1}$, and returns shares of $a$ if $c =1$, else returns shares of $0$ over the same ring. 
A protocol for $\fmux{\ringsz}$ can easily be implemented by 2 simultaneous calls to $\kkot{2}{\etax}$ and communication complexity is $2(\secpar + 2\etax)$, where $\etax = \lceil \log \ringsz \rceil$. 

The functionality $\fBtoA{\ringsz}$ (for boolean to arithmetic conversion) takes boolean (i.e., over $\bbZ_2$) shares as input and gives out arithmetic (i.e., over $\bbZ_{\ringsz}$) shares of the same value as output. It can be realized via one call to $\iknpcot{\etax}$ and hence, its communication is $\secpar+\etax$. For completeness, we provide the protocols realizing $\fmux{\ringsz}$ as well as $\fBtoA{\ringsz}$ formally in \appendixref{mux} and \appendixref{BtoA}, respectively.

\subsubsection{Homomorphic Encryption}
\label{sec:he}

A homomorphic encryption of $x$ allows computing encryption of $f(x)$ without the knowledge of the decryption key. 
In this work, we require an additively homomorphic encryption scheme that supports addition and scalar multiplication, i.e. multiplication of a ciphertext with a plaintext. 
We use the additively homomorphic scheme of BFV~\cite{brak12,fv} (the scheme used in the recent works of Gazelle~\cite{gazelle} and Delphi~\cite{delphi}) and use the optimized algorithms of Gazelle for homomorphic matrix-vector products and homomorphic convolutions. 
The BFV scheme uses the batching optimization \cite{fully-hom-simd-operations,sealmanual} that enables operation on plaintext vectors over the field $\bbZ_{\ringsz}$, where $\ringsz$ is a prime plaintext modulus of the form $2 K \polymod + 1$, $K$ is some positive integer and $\polymod$ is scheme parameter that is a power-of-$2$.

\section{Millionaires' and $\drelu$ protocols}
\label{sec:comp-protocols}
In this section, we provide our protocols for millionaires' problem and $\drelu(a)$ when the inputs are $\ell$ bit signed integers as well as elements in general rings of the form $\bbZ_{\ringsz}$ (including prime fields).
Our protocol for millionaires' problem invokes instances of $\fand$ that take as input boolean shares of values $x, y \in \zo$ and returns boolean shares of $x \wedge y$. We discuss efficient  protocols for $\fand$ in \appendixref{regular-bit-triple} and ~\ref{app:correlated-bit-triple}.

\newcommand{\leavesmill}[2]{\lceil #1/#2 \rceil}
\newcommand{\tworemblsize}{R}

\subsection{Protocol for Millionaires'}
\label{sec:prot-mill}

In the Yao millionaires' problem, party $\party{0}$ holds $x$ and party $\party{1}$ holds $y$ and they wish to learn boolean shares of $\ind{x < y}$. 
Here, $x$ and $y$ are $\ell$-bit unsigned integers.
We denote this functionality by $\fmill{\ell}$.
Our protocol for $\fmill{\ell}$ builds on the following observation that was also used in \cite{GSV07}.
\begin{equation} \label{eq:cmp_eq}
    \ind{x < y} = \ind{x_1 < y_1} \xor \left(\ind{x_1 = y_1} \wedge \ind{x_0 < y_0}\right),
\end{equation}
where, $x = x_1\concat x_0$ and $y = y_1\concat y_0$. \\

\noindent{\em Intuition.} Let $\blsize$ be a parameter and $\twoblsize = 2^\blsize$. First, for ease of exposition, we consider the special case when $\blsize$ divides $\ell$ and $\blnum = \ell/\blsize$ is a power of 2. We describe our protocol for millionaires' problem in this setting formally in \algoref{mill}. We use Equation~\ref{eq:cmp_eq} above,  recursively $\log \blnum$ times to obtain $\blnum$ leaves of size $\blsize$ bits. That is, let $x = x_{\blnum-1} \concat \ldots \concat x_0$ and $y = y_{\blnum-1} \concat \ldots \concat y_0$ (where every $x_i, y_i \in \zo^m$).
Now, we compute the shares of the inequalities and equalities of strings at the leaf level using $\kkot{\twoblsize}{1}$ (steps~\ref{mill-ot1} and \ref{mill-ot2}, resp.).
Next, we compute the shares of the inequalities (steps~\ref{mill-and-comp} \& \ref{mill-xor-comp}) and equalities (step~\ref{mill-and-eq}) at each internal node upwards from the leaf using Equation~\ref{eq:cmp_eq}. Value of inequality at the root gives the final output. \\

\begin{algorithm}[t]
\caption{Millionaires', $\protmill{\ell, \blsize}$:}
\label{algo:mill}
\begin{algorithmic}[1]

\Require $\party{0}, \party{1}$ hold $x \in \zo^\ell$ and $y \in \zo^\ell$, respectively.
    
\Ensure $\party{0},\party{1}$ learn $\share{\ind{x < y}}{B}{0}$ and $\share{\ind{x < y}}{B}{1}$, respectively.

\vspace{0.2cm}

\State \alice parses its input as $x = x_{\blnum-1} \concat \ldots \concat x_0$
        and \bob parses its input as $y = y_{\blnum-1} \concat \ldots \concat y_0$, where $x_i, y_i \in \zo^{\blsize}$, $\blnum = \ell/\blsize$.
\State Let $\twoblsize = 2^\blsize$.
\For{$j = \{0, \ldots, \blnum-1\}$}
            \State \alice samples  $\share{\lessthan_{0,j}}{B}{0}, \share{\equal_{0,j}}{B}{0} \getsr \zo$.
						\For{$k = \{ 0, \ldots, \twoblsize-1 \}$} \label{mill-message-set}
                \State \alice sets $s_{j,k} = \share{\lessthan_{0,j}}{B}{0} \xor \ind{x_j < k}$.
                \State \alice sets $t_{j,k} = \share{\equal_{0,j}}{B}{0} \xor \ind{x_j = k}$.
            \EndFor
            \State \alice \& \bob invoke an instance of $\kkot{\twoblsize}{1}$ where \alice is the sender with inputs $\{s_{j,k}\}_{k}$ and \bob is the receiver with input $y_j$. \bob sets its output as $\share{\lessthan_{0,j}}{B}{1}$. \label{mill-ot1}
            \State \alice \& \bob invoke an instance of $\kkot{\twoblsize}{1}$ where \alice is the sender with inputs $\{t_{j,k}\}_{k}$ and \bob is the receiver with input $y_j$. \bob sets its output as $\share{\equal_{0,j}}{B}{1}$. \label{mill-ot2}
\EndFor

\For {$i = \{1, \ldots, \log \blnum\}$}
		\For {$j = \{0, \ldots, (\blnum/2^i)-1\}$}
				\State For $b \in \zo$, $\party{b}$ invokes $\fand$ with inputs $\share{\lessthan_{i-1, 2j}}{B}{b}$ and $\share{\equal_{i-1, 2j+1}}{B}{b}$ to learn output $\share{\temp}{B}{b}$. \label{mill-and-comp}
				\State $\party{b}$ sets $\share{\lessthan_{i,j}}{B}{b} =  \share{\lessthan_{i-1, 2j+1}}{B}{b} \xor \share{\temp}{B}{b}$. \label{mill-xor-comp}
				\State For $b \in \zo$, $\party{b}$ invokes $\fand$ with inputs $\share{\equal_{i-1,2j}}{B}{b}$ and $\share{\equal_{i-1,2j+1}}{B}{b}$ to learn output $\share{\equal_{i,j}}{B}{b}$. \label{mill-and-eq}
		\EndFor
\EndFor

\State For $b\in\zo$, $\party{b}$ outputs $\share{\lessthan_{\log \blnum, 0}}{B}{b}$.
\end{algorithmic}
\end{algorithm}
 
\noindent{\em Correctness and security.} Correctness is shown by induction on the depth of the tree starting at the leaves. First, by correctness of $\kkot{\twoblsize}{1}$ in step~\ref{mill-ot1}, $\share{\lessthan_{0,j}}{B}{1} = \share{\lessthan_{0,j}}{B}{0} \xor \ind{x_j < y_j}$. Similarly,  $\share{\equal_{0,j}}{B}{1} = \share{\equal_{0,j}}{B}{0} \xor \ind{x_j = y_j}$. This proves the base case. Let $\blnum_i = \blnum/2^i$. Also, for level $i$ of the tree, parse $x = x\p{i} = x\p{i}_{q_i-1}\concat \ldots x\p{i}_0$ and $y = y\p{i} = y\p{i}_{q_i-1}\concat \ldots y\p{i}_0$. Assume that for $i$ it holds that $\lessthan_{i,j} = \share{\lessthan_{i,j}}{B}{0} \xor \share{\lessthan_{i,j}}{B}{1} = \ind{x\p{i}_j < y\p{i}_j}$ and $\share{\equal_{i,j}}{B}{0} \xor \share{\equal_{i,j}}{B}{1} = \ind{x\p{i}_j = y\p{i}_j}$ for all $j \in \{0, \ldots, \blnum_i-1\}$. Then, we prove the same for $i+1$ as follows: By correctness of $\fand$, for $j \in  \{0, \ldots, \blnum_{i+1}-1\}$,  $\share{\lessthan_{i+1,j}}{B}{0} \xor \share{\lessthan_{i+1,j}}{B}{1} = \lessthan_{i, 2j+1} \xor (\lessthan_{i, 2j} \wedge \equal_{i, 2j+1})   = \ind{x\p{i}_{2j+1} < y\p{i}_{2j+1}} \xor (\ind{x\p{i}_{2j} < y\p{i}_{2j}}\wedge \ind{x\p{i}_{2j+1} = y\p{i}_{2j+1}})  =  \ind{x\p{i+1}_j < y\p{i+1}_j}$ (using Equation~\ref{eq:cmp_eq}). The induction step for $\equal_{i+1, j}$ holds in a similar manner, thus proving correctness. Given uniformity of $\share{\lessthan_{0,j}}{B}{0}, \share{\equal_{0,j}}{B}{0}$ for all $j \in \{0, \ldots, \blnum-1\}$, security follows easily in the $(\kkot{\twoblsize}{1}, \fand)$-hybrid. \\

\noindent{\em General case.} When $\blsize$ does not divide $\ell$ and $\blnum = \leavesmill{\ell}{\blsize}$ is not a power of 2,  we make the following modifications to the protocol. 
Since $\blsize$ does not divide $\ell$, $x_{\blnum-1} \in \zo^{r}$, where $r = \ell \bmod \blsize$.\footnote{Note that $r = \blsize$ when $\blsize$ divides $\ell$.} 
When doing the compute for $x_{\blnum -1}$ and $y_{\blnum-1}$, we perform a small optimization and use $\kkot{\tworemblsize}{1}$ in steps \ref{mill-ot1} and \ref{mill-ot2}, where $\tworemblsize = 2^{r}$.
Second, since $\blnum$ is not a power of 2, we do not have a perfect binary tree of recursion and we need to slightly change our recursion/tree traversal.
In the general case, we construct maximal possible perfect binary trees and connect the roots of the same using the relation in Equation~\ref{eq:cmp_eq}.
Let $\alpha$ be such that $2^\alpha < \blnum \leq 2^{\alpha+1}$. 
Now, our tree has a perfect binary sub-tree with $2^\alpha$ leaves and we have remaining $\blnum' = \blnum-2^\alpha$ leaves. We recurse on $\blnum'$. 
In the last step, we obtain our tree with $\blnum$ leaves by combining the roots of perfect binary tree with $2^\alpha$ leaves and tree with $\blnum'$ leaves using Equation~\ref{eq:cmp_eq}. 
Note that value at the root is computed using $\lceil \log \blnum \rceil$  sequential steps starting from the leaves.

\subsubsection{Optimizations}
We reduce the concrete communication complexity of our protocol using the following optimizations that are applicable to both the special and the general case.
\begin{tiret}
	\item Combining two $\kkot{\twoblsize}{1}$ calls into one $\kkot{\twoblsize}{2}$: Since the input of \bob (OT receiver) to $\kkot{\twoblsize}{1}$ in steps \ref{mill-ot1} and \ref{mill-ot2} is same, i.e. $y_j$, we can collapse these steps into a single call to $\kkot{\twoblsize}{2}$ where \alice and \bob input $\{(s_{j,k} \concat t_{j,k})\}_{k}$ and $y_j$, respectively. \bob sets its output as $(\share{\lessthan_{0,j}}{B}{1} \concat \share{\equal_{0,j}}{B}{1})$.  This reduces the cost from $2(2\secpar+\twoblsize)$ to $(2\secpar+2\twoblsize)$.

\item \label{optimill:and} Realizing $\fand$ efficiently:
It is well-known that $\fand$ can be realized using Beaver bit triples~\cite{beaver}.
For our protocol, we observe that the 2 calls to $\fand$ in steps \ref{mill-and-comp} and \ref{mill-and-eq} have a common input, $\share{\equal_{i-1, 2j+1}}{B}{b}$.
Hence, we optimize communication of these steps by generating correlated bit triples $(\share{d}{B}{b}$, $\share{e}{B}{b}$, $\share{f}{B}{b})$ and $(\share{{d'}}{B}{b}$, $\share{e}{B}{b}$, $\share{{f'}}{B}{b})$, for $b \in \zo$, such that $d \wedge e = f$ and ${d'} \wedge e = {f'}$. Next, we use $\kkot{8}{2}$ to generate one such correlated bit triple (\appendixref{correlated-bit-triple}) with communication $2\secpar + 16$ bits, giving the amortized cost of $\secpar + 8$ bits per triple. Given correlated bit triples, we need $6$ additional bits to compute both $\fand$ calls.

	\item \label{optimill:eq} Removing unnecessary equality computations: 
As observed in \cite{GSV07}, the equalities computed on lowest significant bits are never used. Concretely, we can skip computing the values $\equal_{i, 0}$ for $i \in \{0, \ldots, \log \blnum\}$. Once we do this optimization, we only need a single call to $\fand$ instead of 2 correlated calls for the leftmost branch of the tree. 
We use the $\kkot{16}{2} \to 2 \times \kkot{4}{1}$ reduction to generate 2 regular bit triples from \cite{DKSSZZ17} (\appendixref{regular-bit-triple}) with communication of $2\secpar+32$ bits. This gives us amortized communication of $\secpar+16$ bits per triple and we need 4 additional bits to realize $\fand$.
Overall, we get a reduction in total communication by $\twoblsize$ (for the leaf) plus $(\secpar + 2) \cdot \lceil \log \blnum \rceil$ (for leftmost branch) bits.

\end{tiret}

\subsubsection{Communication Complexity}
\label{sec:comm-mill}
In our protocol, we communicate in protocols for $\mathsf{OT}$ (steps~\ref{mill-ot1}\&\ref{mill-ot2}) and $\fand$ (steps~\ref{mill-and-comp}\&\ref{mill-and-eq}). 
With above optimizations, we need 1 call to $\kkot{\twoblsize}{1}$, $(\blnum-2)$ calls to $\kkot{\twoblsize}{2}$ and 1 call to $\kkot{\tworemblsize}{2}$ which cost $(2\secpar + \twoblsize)$, $\big((\blnum-2) \cdot (2\secpar + 2\twoblsize)\big)$ and $(2\secpar + 2\tworemblsize)$ bits, respectively. 
In addition, we have $\lceil \log \blnum \rceil$ invocations of $\fand$ and $(\blnum - 1 - \lceil \log \blnum \rceil)$ invocations of correlated $\fand$. These require communication of $(\secpar+20)\cdot\lceil \log \blnum \rceil$ and $(2\secpar+22)\cdot(\blnum - 1 - \lceil \log \blnum \rceil)$ bits. This gives us total communication of 
 $\secpar(4\blnum - \lceil \log \blnum \rceil -2) + \twoblsize(2\blnum - 3) + 2\tworemblsize + 22(\blnum-1) -2\lceil \log \blnum \rceil$ bits.
Using this expression for $\ell=32$ we get least communication for $\blsize = 7$ (\tableref{comp-mill}). We note that there is a trade-off between communication and computational cost of $\mathsf{OT}$s used and we discuss our choice of $\blsize$ for our experiments in \sectionref{impl}.

\subsection{Protocol for $\drelu$ for $\ell$-bit integers}
\label{sec:prot-drelu-int}
In \algoref{drelu-int}, we describe our protocol for $\fdreluint{\ell}$ that takes as input arithmetic shares of $a$ and returns boolean shares of $\drelu(a)$. 
Note that $\drelu(a) = (1\xor \msbl(a))$, where $\msbl(a)$ is the most significant bit of $a$. 
Let arithmetic shares of $a \in \bbZ_{\intx}$ be $\share{a}{\intx}{0} = \msbs_0\concat x_0$ and $\share{a}{\intx}{1} = \msbs_1\concat x_1$ such that $\msbs_0, \msbs_1 \in \zo$. 
 We compute the boolean shares of $\msbl(a)$ as follows: Let $\carry = \ind{(x_0 + x_1) > 2^{\ell-1}-1}$. Then, $\msbl(a) = \msbs_0 \xor \msbs_1 \xor \carry$. We compute boolean shares of carry by invoking an instance of $\fmill{\ell-1}$. \\

\begin{algorithm}[t]
\caption{$\ell$-bit integer $\drelu$, $\protdreluint{\ell}$:}
\label{algo:drelu-int}
\begin{algorithmic}[1]

 \Require $\party{0}, \party{1}$ hold $\share{a}{\intx}{0}$ and $\share{a}{\intx}{1}$, respectively. 
\Ensure $\party{0}, \party{1}$ get $\share{\drelu(a)}{B}{0}$ and $\share{\drelu(a)}{B}{1}$.

\vspace{0.2cm}

\State \alice parses its input as $\share{a}{\intx}{0} = \msbs_0\concat  x_0$
        and \bob parses its input as $\share{a}{\intx}{1} = \msbs_1 \concat x_1$, s.t. $b\in \zo, \msbs_b \in \zo, x_b \in  \zo^{\ell-1}$.

\State \alice \& \bob invoke an instance of $\fmill{\ell-1}$, 
where \alice's input is $2^{\ell-1}-1 -x_0$ and \bob's input is $x_1$. For $b\in \zo$, $\party{b}$ learns $\share{\carry}{B}{b}$.   \label{relu-int-mill}

\State For $b \in \zo$, $\party{b}$ sets $\share{\drelu}{B}{b} = \msbs_b \xor \share{\carry}{B}{b} \xor b$.

\end{algorithmic}
\end{algorithm}
 
\noindent{\em Correctness and security.} By correctness of $\fmill{\ell-1}$, $\reconst{B}{\share{\carry}{B}{0},$\\ $\share{\carry}{B}{1}} = \ind{(2^{\ell-1}-1-x_0) < x_1} = \ind{(x_0 + x_1) > 2^{\ell-1}-1}$. Also, $\reconst{B}{\share{\drelu}{B}{0}, \share{\drelu}{B}{1}} = \msbs_0 \xor \msbs_1 \xor \carry \xor 1 = \msbl(a) \xor 1$. Security follows trivially in the $\fmill{\ell-1}$ hybrid.  
\\

\noindent{\em Communication complexity} 
In \algoref{drelu-int}, we communicate the same as in $\protmill{\ell-1}$, that is $< (\secpar+14)(\ell-1)$ by using $\blsize=4$. 

\subsection{Protocol for $\drelu$ for general $\bbZ_{\ringsz}$}
\label{sec:prot-drelu-ring}
We describe a protocol for $\fdreluring{\ringsz}$ that takes arithmetic shares of $a$ over $\bbZ_\ringsz$ as input and returns boolean shares of $\drelu(a)$.
For integer rings $\bbZ_{\ringsz}$, $\drelu(a) = 1$ if $a < \lceil \ringsz/2 \rceil$ and $0$ otherwise. 
Note that this includes the case of prime fields considered in the works of \cite{gazelle,delphi}. 
Below, we formally discuss the case of rings of odd number of elements and omit the analogous case of even rings. 
We first describe a (simplified) protocol for $\drelu$ over $\bbZ_{\ringsz}$ in \algoref{drelu-ring-simple} with protocol logic as follows: Let arithmetic shares of $a\in\bbZ_{\ringsz}$ be $\share{a}{\ringsz}{0}$ and $\share{a}{\ringsz}{1}$. Define $\wrap = \ind{\share{a}{\ringsz}{0} + \share{a}{\ringsz}{1} > \ringsz-1}$, $\lt = \ind{\share{a}{\ringsz}{0} + \share{a}{\ringsz}{1} > (\ringsz-1)/2}$ and $\rt= \ind{\share{a}{\ringsz}{0} + \share{a}{\ringsz}{1} > \ringsz + (\ringsz-1)/2}$. Then, $\drelu(a)$ is $(1\xor \lt)$ if $\wrap = 0$, else it is $(1\xor \rt)$. In \algoref{drelu-ring-simple}, steps~\ref{relu-ring-mill-one},\ref{relu-ring-mill-two},\ref{relu-ring-mill-three}, compute these three comparisons using $\fmill{}$. Final output can be computed using an invocation of $\fmux{2}$. 

\begin{algorithm}[t]
\caption{Simple Integer ring $\drelu$, $\protdreluringsimple{\ringsz}$:}
\label{algo:drelu-ring-simple}
\begin{algorithmic}[1]

\Require $\party{0}, \party{1}$ hold $\share{a}{\ringsz}{0}$ and $\share{a}{\ringsz}{1}$, respectively, where $a \in \bbZ_{\ringsz}$. 
\Ensure $\party{0}, \party{1}$ get $\share{\drelu(a)}{B}{0}$ and $\share{\drelu(a)}{B}{1}$.

\vspace{0.2cm}

\State \alice \& \bob invoke an instance of $\fmill{\etax}$ with $\etax = \lceil \log \ringsz \rceil$, 
where \alice's input is $\left(\ringsz-1 - \share{a}{\ringsz}{0}\right)$ and \bob's input is $\share{a}{\ringsz}{1}$. For $b\in \zo$, $\party{b}$ learns $\share{\wrap}{B}{b}$ as output.  \label{relu-ring-mill-one}

\State \alice \& \bob invoke an instance of $\fmill{\etax+1}$, 
where \alice's input is $\left(\ringsz-1 - \share{a}{\ringsz}{0}\right)$ and \bob's input is $\left((\ringsz-1)/2 + \share{a}{\ringsz}{1}\right)$. For $b\in \zo$, $\party{b}$ learns $\share{\lt}{B}{b}$ as output.  \label{relu-ring-mill-two}

\State \alice \& \bob invoke an instance of $\fmill{\etax+1}$, 
where \alice's input is $\left(\ringsz + (\ringsz-1)/2 - \share{a}{\ringsz}{0}\right)$ and \bob's input is $\share{a}{\ringsz}{1}$. For $b\in \zo$, $\party{b}$ learns $\share{\rt}{B}{b}$ as output.  \label{relu-ring-mill-three}

	\State For $b\in \zo$, $\party{b}$ invokes $\fmux{2}$  with input $\left(\share{\lt}{B}{b} \xor \share{\rt}{B}{b}\right)$ and choice $\share{\wrap}{B}{b}$ to learn $\share{z}{B}{b}$. 
\State For $b \in \zo$, $\party{b}$ outputs $\share{z}{B}{b} \xor \share{\lt}{B}{b} \xor b$.

\end{algorithmic}
\end{algorithm}

 \begin{algorithm}[t]
\caption{Optimized Integer ring $\drelu$, $\protdreluring{\ringsz}$:}
\label{algo:drelu-ring}
\begin{algorithmic}[1]

\Require $\party{0}, \party{1}$ hold $\share{a}{\ringsz}{0}$ and $\share{a}{\ringsz}{1}$, respectively, where $a \in \bbZ_{\ringsz}$. Let $\etax = \lceil \log \ringsz \rceil$. 
\Ensure $\party{0}, \party{1}$ get $\share{\drelu(a)}{B}{0}$ and $\share{\drelu(a)}{B}{1}$.

\vspace{0.2cm}

\State \alice \& \bob invoke an instance of $\fmill{\etax+1}$, 
	where \alice's input is $\left(3(\ringsz-1)/2 - \share{a}{\ringsz}{0}\right)$ and \bob's input is $(\ringsz-1)/2 + \share{a}{\ringsz}{1}$. For $b\in \zo$, $\party{b}$ learns $\share{\wrap}{B}{b}$ as output. \label{relu-ring-final-mill-one}

	\State \alice sets $x = \left(2n -1 - \share{a}{\ringsz}{0}\right)$ if $\share{a}{\ringsz}{0} > (\ringsz-1)/2$, else $x = \left(\ringsz -1 - \share{a}{\ringsz}{0}\right)$.

\State \alice \& \bob invoke an instance of $\fmill{\etax+1}$, 
where \alice's input is $x$  and \bob's input is $\left((\ringsz-1)/2 + \share{a}{\ringsz}{1}\right)$. For $b\in \zo$, $\party{b}$ learns $\share{\xt}{B}{b}$ as output.  \label{relu-ring-final-mill-two}

	\State \alice samples  $\share{z}{B}{0} \getsr \zo$.

\For{$j = \{00, 01, 10, 11\}$}
						\State \alice parses $j$ as $j_{0} \concat j_{1}$ and sets $t_{j} = 1 \xor \share{\xt}{B}{0} \xor j_0$. 
						\If{$\share{a}{\ringsz}{0} > (\ringsz-1)/2$}
							\State \alice sets $s'_{j} = t_{j} \wedge (\share{\wrap}{B}{0} \xor j_1)$.
						\Else
							\State \alice sets $s'_{j} = t_{j} \xor ((1 \xor t_{j}) \wedge (\share{\wrap}{B}{0} \xor j_1))$
						\EndIf
						\State \alice sets $s_{j} = s'_{j} \xor \share{z}{B}{0}$
\EndFor

\State \alice \& \bob invoke an instance of $\kkot{4}{1}$ where \alice is the sender with inputs $\{s_j\}_{j}$ and \bob is the receiver with input $\share{\xt}{B}{1} \concat \share{\wrap}{B}{1}$. \bob sets its output as $\share{z}{B}{1}$. 

\State For $b \in \zo$, $\party{b}$ outputs $\share{z}{B}{b}$.

\end{algorithmic}
\end{algorithm}
 
\noindent{\em Optimizations.} 
We describe an optimized protocol for $\fdreluring{\ringsz}$ in \algoref{drelu-ring} that reduces the number of calls to $\fmill{}$ to 2. 
First, we observe that if the input of \bob is identical in all three invocations, then the invocations of OT in \algoref{mill} (steps~\ref{mill-ot1}$\&$\ref{mill-ot2}) can be done together for the three comparisons. This reduces the communication for each leaf OT invocation in steps~\ref{mill-ot1}$\&$\ref{mill-ot2} by an additive factor of $4\lambda$. To enable this, $\party{0}, \party{1}$ add $(\ringsz-1)/2$ to their inputs to $\fmill{\etax+1}$ in steps~\ref{relu-ring-mill-one},\ref{relu-ring-mill-three} ($\etax = \lceil \log \ringsz \rceil$). Hence, \bob's input to $\fmill{\etax+1}$ is $(\ringsz-1)/2 + \share{a}{\ringsz}{1}$ in all invocations and \alice's inputs are $\left(3(\ringsz-1)/2 - \share{a}{\ringsz}{0}\right)$, $\left(\ringsz -1 - \share{a}{\ringsz}{0}\right)$, $\left(2n-1 - \share{a}{\ringsz}{0}\right)$ in steps~\ref{relu-ring-mill-one},\ref{relu-ring-mill-two},\ref{relu-ring-mill-three}, respectively. 

Next, we observe that one of the comparisons in step \ref{relu-ring-mill-two} or step~\ref{relu-ring-mill-three} is redundant. For instance, if $\share{a}{\ringsz}{0} > (\ringsz-1)/2$, then the result of the comparison $\lt = \share{a}{\ringsz}{0} + \share{a}{\ringsz}{1} > (\ringsz-1)/2$ done in step~\ref{relu-ring-mill-two} is always $1$. Similarly, if $\share{a}{\ringsz}{0} \leq (\ringsz-1)/2$, then the result of the comparison $\rt= \ind{\share{a}{\ringsz}{0} + \share{a}{\ringsz}{1} > \ringsz + (\ringsz-1)/2}$ done in step~\ref{relu-ring-mill-three} is always $0$. Moreover, \alice knows based on her input $\share{a}{\ringsz}{0}$ which of the two comparisons is redundant. Hence, in the optimized protocol, \alice and \bob always run the comparison to compute shares of $\wrap$ and one of the other two comparisons. Note that the choice of which comparison is omitted by \alice need not be communicated to \bob, since \bob's input is same in all invocations of $\fmill{}$. Moreover, this omission does not reveal any additional information to \bob by security of $\fmill{}$.
Finally, \alice and \bob can run a $\kkot{4}{1}$ to learn the shares of $\drelu(a)$.  Here, \bob is the receiver and her choice bits are the shares learnt in the two comparisons. \alice is the sender who sets the 4 OT messages based on her input share, and two shares learnt from the comparison protocol. We elaborate on this in the correctness proof below. \\

\noindent{\em Correctness and Security.} First, by correctness of $\fmill{\etax+1}$ (step~\ref{relu-ring-final-mill-one}), $\wrap = \reconst{B}{\share{\wrap}{B}{0}, \share{\wrap}{B}{1}} = \ind{\share{a}{\intx}{0} + \share{a}{\intx}{1} > \ringsz -1}$. 
 Let $j^* = \share{\xt}{B}{1} \concat \share{\wrap}{B}{1}$. Then, $t_{j^*} = 1 \xor \xt$. We will show that $s'_{j^*} = \drelu(a)$, and hence, by correctness of $\kkot{4}{1}$, $z = \reconst{B}{\share{z}{B}{0}, \share{z}{B}{1}} = \drelu(a)$. We have the following two cases. 

When $\share{a}{\intx}{0} > (\ringsz-1)/2$, $\lt = 1$, and $\drelu(a) = \wrap \wedge (1 \xor \rt)$. Here, by correctness of $\fmill{\etax+1}$ (step~\ref{relu-ring-final-mill-two}), $\xt = 
\reconst{B}{\share{\xt}{B}{0}, \share{\xt}{B}{1}} = \rt$. Hence, $s'_{j^*} = t_{j^*} \wedge (\share{\wrap}{B}{0} \xor j^*_1) = (1\xor \rt) \wedge \wrap$.

When $\share{a}{\intx}{0} \leq (\ringsz-1)/2$, $\rt = 0$, $\drelu(a) $ is $1\xor \lt$ if $\wrap = 0$, else $1$. It can be written as $(1\xor \lt) \xor (\lt \wedge \wrap)$. In this case, by correctness of $\fmill{\etax+1}$ (step~\ref{relu-ring-final-mill-two}), $\xt = 
\reconst{B}{\share{\xt}{B}{0}, \share{\xt}{B}{1}} = \lt$. Hence, $s'_{j^*} = t_{j^*} \xor ((1 \xor t_{j^*}) \wedge (\share{\wrap}{B}{0} \xor j^*_1)) = (1\xor \lt) \xor (\lt \wedge \wrap)$. Since $\share{z}{B}{0}$ is uniform, security follows in the $(\fmill{\etax+1}, \kkot{4}{1})$-hybrid. \\

\noindent{\em Communication complexity.}
With the above optimization, the overall communication complexity of our protocol for $\drelu$ in $\bbZ_{\ringsz}$ is equivalent to $2$ calls to $\protmill{\etax+1}$ where \bob has same input plus $2\secpar+4$ (for protocol for $\kkot{4}{1}$).
Two calls to $\protmill{\etax+1}$ in this case (using $\blsize = 4$) cost $< \frac{3}{2} \secpar (\etax + 1) + 28 (\etax + 1)$ bits.
Hence, total communication is $< \frac{3}{2} \secpar (\etax + 1) + 28 (\etax + 1) + 2 \secpar + 4$.
We note that the communication complexity of simplified protocol in \algoref{drelu-ring-simple} is approximately $3$ independent calls to $\protmill{\eta}$, which cost $3 (\secpar \etax + 14 \etax)$ bits, plus $2\secpar+4$ bits for $\fmux{2}$.
Thus, our optimization gives almost $2\times$ improvement.

\section{Division and truncation} 
\label{sec:division}

We present our results on secure implementations of division in the ring by a positive integer and truncation (division by power-of-$2$) that are bitwise equivalent to the corresponding cleartext computation. We begin with closed form expressions for each of these followed by secure protocols that use them.

\subsection{Expressing general division and truncation using arithmetic over secret shares}
\label{sec:express-div}

Let $\idiv: \bbZ \times \bbZ \rightarrow \bbZ$ denote signed integer division, where the quotient is rounded towards $- \infty$ and the sign of the remainder is the same as that of divisor.
We denote division of a ring element by a positive integer using $\rdiv: \bbZ_{\ringsz} \times \bbZ \rightarrow \bbZ_{\ringsz}$ defined as
\begin{equation} \label{eq:div-eqn}
    \frdiv{a}{d} \triangleq \fidiv{a_u - \ind{a_u \geq \lceil \ringsz/2 \rceil} \cdot \ringsz}{d} \bmod{\ringsz},
\end{equation}
where the integer $a_u \in \{0,1,\ldots, n-1\}$ is the unsigned representation of $a \in \bbZ_{\ringsz}$ lifted to integers and $0 < d < \ringsz$.
For brevity, we use $x =_{\ringsz} y$ to denote $x \bmod{\ringsz} = y \bmod{\ringsz}$.

\begin{theorem} 
\label{theorem:general-division}
(Division of ring element by positive integer).
Let the shares of $a \in \bbZ_{\ringsz}$ be $\share{a}{\ringsz}{0}, \share{a}{\ringsz}{1} \in \bbZ_{\ringsz}$, for some $\ringsz = \ringsz^1 \cdot d + \ringsz^0 \in \bbZ$, where $\ringsz^0, \ringsz^1, d \in \bbZ$ and $0 \leq \ringsz^0 < d < \ringsz$.

Let the unsigned representation of $a,\share{a}{\ringsz}{0}, \share{a}{\ringsz}{1}$ in $\bbZ_{\ringsz}$ lifted to integers be $a_u, a_0, a_1 \in \{0, 1, \ldots, n-1\}$, respectively, such that $a_0 = a_0^1 \cdot d + a_0^0$ and $a_1 = a_1^1 \cdot d + a_1^0$, where $a_0^1, a_0^0, a_1^1, a_1^0 \in \bbZ$ and $0 \leq a_0^0, a_1^0 < d$.
Let $\ringsz' = \lceil \ringsz/2 \rceil \in \bbZ$. Define $\corr,\ A, \ B,\ C \in \bbZ$ as follows:
\begin{equation*}
\begin{split}
    \corr ={}& \left \{
    \begin{array}{cc}
        -1 & (a_u \geq \ringsz') \wedge (a_0 < \ringsz') \wedge (a_1 < \ringsz') \\
         1 & (a_u < \ringsz') \wedge (a_0 \geq \ringsz') \wedge (a_1 \geq \ringsz') \\
         0 & \text{otherwise}
    \end{array} \right . , \\
    A ={}& a_0^0 + a_1^0 - (\ind{a_0 \geq \ringsz'} + \ind{a_1 \geq \ringsz'} - \corr) \cdot \ringsz^0. \\
    B = {} & \fidiv{a_0^0 - \ind{a_0 \geq \ringsz'} \cdot \ringsz^0}{d} + \fidiv{a_1^0 - \ind{a_1 \geq \ringsz'} \cdot \ringsz^0}{d} \\
    C = {} & \ind{A~<~d} + \ind{A~<~0} + \ind{A~<~-d}
\end{split}
\end{equation*}
Then, we have:
\begin{equation*}
    \frdiv{\share{a}{\ringsz}{0}}{d} + \frdiv{\share{a}{\ringsz}{1}}{d} + (\corr \cdot \ringsz^1 + 1 - C - B) =_{\ringsz} \frdiv{a}{d}.
\end{equation*}
\end{theorem}

The proof of the above theorem is presented in \appendixref{division-proof}.

\subsubsection{Special Case of truncation for $\ell$ bit integers} 
The expression  above  can be simplified for the special case of division by $2^s$  of $\ell$-bit integers, i.e., arithmetic right shift with $s$ ($\gg s$), as follows: 

\begin{corollary} 
\label{corol:truncate-int}
(Truncation for $\ell$-bit integers).
Let the shares of $a \in \bbZ_{\intx}$ be $\share{a}{\intx}{0}, \share{a}{\intx}{1} \in \bbZ_{\intx}$.
Let the unsigned representation of $a,\share{a}{\intx}{0}, \share{a}{\intx}{1}$ in $\bbZ_{\intx}$ lifted to integers be $a_u, a_0, a_1 \in \{0, 1, \ldots, 2^\ell-1\}$, respectively, such that $a_0 = a_0^1 \cdot 2^s + a_0^0$ and $a_1 = a_1^1 \cdot 2^s + a_1^0$, where $a_0^1, a_0^0, a_1^1, a_1^0 \in \bbZ$ and $0 \leq a_0^0, a_1^0 < 2^s$.
Let $\corr \in \bbZ$ be defined as in \theoremref{general-division}.
Then, we have:
\begin{equation*}
    (a_0 \gg s)  + (a_1 \gg s) + \corr \cdot 2^{\ell-s} + \ind{a_0^0 + a_1^0 \geq 2^s} =_\intx (a \gg s).
\end{equation*}
\end{corollary}

\begin{proof}
The corollary follows directly from \theoremref{general-division} as follows: First, $(a \gg s) = \frdiv{a}{2^s}$. Next, $\ringsz = 2^\ell$, $\ringsz^1 = 2^{\ell-s}$, and $\ringsz^0 = 0$. Using these, we get $A =a_0^0 + a_1^0$, $B = 0$ and $C = \ind{A < 2^s} = \ind{a_0^0 + a_1^0 < 2^s}$. 
\end{proof}

\subsection{Protocols for division}
\label{sec:division-protocols}

In this section, we describe our protocols for division in different settings.
We first describe a protocol for the simplest case of truncation for $\ell$-bit integers followed by a protocol for general division in $\bbZ_\ringsz$ by a positive integer (\sectionref{prot-general-division}).
Finally, we discuss another simpler case of truncation, which allows us to do better than general division for rings with a special structure (\sectionref{truncate-special-rings}).

\subsubsection{\textbf{Protocol for truncation of $\ell$-bit integer}}
\label{sec:prot-truncate-int}

\begin{algorithm}
\caption{Truncation,  $\prottruncint{\ell}{s}$:}
\label{algo:truncate-int}
\begin{algorithmic}[1]
\Require For $b \in \zo$,  $\party{b}$ holds $\share{a}{\intx}{b}$, where $a \in \bbZ_{\intx}$.
\Ensure For $b \in \zo$, $\party{b}$ learns $\share{z}{\intx}{b}$ s.t. $z = a \gg s$.

\vspace{0.2cm}

\State For $b \in \zo$, let $a_b, a_b^0, a_b^1 \in \bbZ$ be as defined in \corolref{truncate-int}.

\State For $b \in \zo$, $\party{b}$ invokes $\fdreluint{\ell}$ with input $\share{a}{\intx}{b}$ to learn output $\share{\alpha}{B}{b}$. Party $\party{b}$ sets $\share{m}{B}{b} = \share{\alpha}{B}{b} \xor b$.

\State For $b \in \zo$, $\party{b}$ sets $x_b = \msbl(\share{a}{\intx}{b})$.
\State \alice samples $\share{\corr}{\intx}{0} \getsr \bbZ_{2^\ell}$.
\For{$j = \{00, 01, 10, 11\}$}
        \State \alice computes $t_{j} = (\share{m}{B}{0} \xor j_0 \xor x_0) \wedge (\share{m}{B}{0} \xor j_0 \xor j_1)$ s.t. $j = (j_0 \concat j_1)$.
    \If{$t_j \wedge \ind{x_0 = 0}$}  \label{corr-case-one}
        \State \alice sets $s_{j} =_{\intx} -\share{\corr}{\intx}{0}-1$.
    \ElsIf{$t_j \wedge \ind{x_0 = 1}$} \label{corr-case-two}
        \State \alice sets $s_{j} =_{\intx} -\share{\corr}{\intx}{0}+1$.
	\Else
        \State \alice sets $s_{j} =_{\intx} -\share{\corr}{\intx}{0}$.
	\EndIf
\EndFor

\State \alice \& \bob invoke an instance of $\kkot{4}{\ell}$, where \alice is the sender with inputs $\{s_{j}\}_j$ and \bob is the receiver with input $\share{m}{B}{1} \concat x_1$ and learns $\share{\corr}{\intx}{1}$.

\State \alice \& \bob invoke an instance of $\fmill{s}$ with \alice's input as $2^s-1-a_0^0$ and \bob's input as $a^0_1$. For $b \in \zo$, $\party{b}$ learns $\share{c}{B}{b}$.

\State  For $b \in \zo$, $\party{b}$ invokes an instance of $\fBtoA{L}$ ($L = 2^\ell$)  with input $\share{c}{B}{b}$ and learns $\share{d}{\intx}{b}$.
\label{bool-arith}

\State $\party{b}$ outputs  $\share{z}{\intx}{b} = (\share{a}{\intx}{b} \gg s) + \share{\corr}{\intx}{b} \cdot 2^{\ell-s} + \share{d}{\intx}{b}$, $b\in \zo$.

\end{algorithmic}
\end{algorithm}
 
Let $\ftruncint{\ell}{s}$ be the functionality that takes arithmetic shares of $a$ as input and returns arithmetic shares of $a \gg s$ as output.
In this work, we give a protocol (\algoref{truncate-int}) that realizes the functionality $\ftruncint{\ell}{s}$ {\em correctly} building on \corolref{truncate-int}.

\noindent{\em Intuition.}
Parties \alice \& \bob first invoke an instance of $\fdreluint{\ell}$ (where one party locally flips its share of $\drelu(a)$) to get boolean shares $\share{m}{B}{b}$ of $\msbl(a)$.
Using these shares, they use a $\kkot{4}{\ell}$ for calculating $\share{\corr}{\intx}{b}$, i.e., arithmetic shares of $\corr$ term in \corolref{truncate-int}. Next, they use an instance of $\fmill{s}$ to compute boolean shares of $c = \ind{a_0^0+ a_1^0 \geq 2^s}$. Finally, they compute arithmetic shares of $c$ using a call to $\fBtoA{L}$ (\algoref{B2A}).\\

\noindent{\em Correctness and Security.} 
For any $z \in \bbZ_{\intx}$,  $\msbl(z) = \ind{z_u \geq 2^{\ell-1}}$, where $z_u$ is unsigned representation of $z$ lifted to integers.
First, note that $\reconst{B}{\share{m}{B}{0}, \share{m}{B}{1}} = 1 \xor \reconst{B}{\share{\alpha}{B}{0}, \share{\alpha}{B}{1}} = \msbl(a)$ by correctness of $\fdreluint{\ell}$.
Next, we show that $\reconst{\intx}{\share{\corr}{\intx}{0}, \\ \share{\corr}{\intx}{1}} = \corr$, as defined in \corolref{truncate-int}. Let $x_b = \msbl(\share{a}{\intx}{b})$ for $b \in \zo$, and let $j^* = (\share{m}{B}{1} \concat x_1)$. Then, $t_{j^*} = (\share{m}{B}{0} \xor \share{m}{B}{1} \xor x_0) \wedge (\share{m}{B}{0} \xor \share{m}{B}{1} \xor x_1) = (\msbl(a) \xor x_0) \wedge (\msbl(a) \xor x_1)$. Now, $t_{j^*} = 1$ implies that we are in one of the first two cases of expression for $\corr$ -- which case we are in can be checked using $x_0$ (steps ~\ref{corr-case-one} \& ~\ref{corr-case-two}). Now it is easy to see that $s_{j^*} = -\share{\corr}{\intx}{0} + \corr = \share{\corr}{\intx}{1}$. 

Next, by correctness of $\fmill{s}$, $c = \reconst{B}{\share{c}{B}{0}, \share{c}{B}{1}} = \share{c}{B}{0} \xor \share{c}{B}{1} = \ind{a_0^0 + a_1^0 \geq 2^s}$. Given boolean shares of $c$, step~\ref{bool-arith}, creates arithmetic shares of the same using an instance of $\fBtoA{L}$.
Since $\share{\corr}{\intx}{0}$ is uniformly random, security of our protocol is easy to see in $(\fdreluint{\ell}, \kkot{4}{\ell}, \fmill{s}, \fBtoA{L})$-hybrid. \\

\noindent{\em Communication complexity.}
$\prottruncint{\ell}{s}$ involves a single call each to $\fdreluint{\ell}, \kkot{4}{\ell}$,  $\fBtoA{L}$ and $\fmill{s}$. Hence, communication required is $<\secpar\ell+2\secpar+19\ell+$ communication for $\fmill{s}$ that depends on parameter $s$. 
For $\ell=32$ and $s=12$, our concrete communication is $4310$ bits (using $\blsize = 7$ for $\protmill{12}$ as well as $\protmill{31}$ inside $\protdreluint{32}$) as opposed to 24064 bits for garbled circuits.

\subsubsection{\textbf{Protocol for division in ring}}
\label{sec:prot-general-division}
Let $\fdivring{n}{d}$ be the functionality for division that takes arithmetic shares of $a$ as input and returns arithmetic shares of $\frdiv{a}{d}$ as output. 
Our protocol builds on our closed form expression from \theoremref{general-division}. 
We note that $\ell$-bit integers is a special case of $\bbZ_\ringsz$ and we use the same protocol for division of an element in $\bbZ_\intx$ by a positive integer.\\

\noindent{\em Intuition.} This protocol is similar to the previous protocol for truncation and uses the same logic to compute shares of $\corr$ term. 
Most non-trivial term to compute is $C$ that involves three signed comparisons over $\bbZ$. We emulate these comparisons using calls to $\fdreluint{\delta}$ where $\delta$ is large enough to ensure that there are no overflows or underflows. 
It is not too hard to see that $-2d+2 \leq A \leq 2d-2$ and hence, $-3d+2 \leq A-d,A, A+d \leq 3d-2$. Hence, we set $\delta = \lceil \log 6d \rceil$.
Now, with this value of $\delta$, the term $C$ can we re-written as $(\drelu(A- d) \xor 1) + (\drelu(A) \xor 1) + (\drelu(A+ d) \xor 1)$, which can be computed using three calls to $\fdreluint{\delta}$ (Step~\ref{div-drelu}) and $\fBtoA{\ringsz}$ (Step~\ref{div-B2A}) each. Finally, note that to compute $C$ we need arithmetic shares of $A$ over the ring $\bbZ_\Delta$, $\Delta = 2^\delta$. And this requires shares of $\corr$ over the same ring. Hence, we compute shares of $\corr$ over both $\bbZ_\ringsz$ and $\bbZ_\Delta$ (Step~\ref{div-corr}). Due to space constraints, we describe the protocol formally in \appendixref{division-protocol} along with its communication complexity. Also, \tableref{comp-avgpool} provides theoretical and concrete communication numbers for division in both $\bbZ_\intx$ and $\bbZ_\ringsz$, as well as a comparison with garbled circuits.\\

\subsubsection{\textbf{Truncation in rings with special structure}}
\label{sec:truncate-special-rings}
It is easy to see that truncation by $s$ in general rings can be done by performing a division by $d = 2^s$.
However, we can omit a call to $\fdreluint{\delta}$ and $\fBtoA{\ringsz}$ when the underlying ring and $d$ satisfy a relation.
Specifically, if we have $2 \cdot \ringsz^0 \leq d = 2^s$, then $A$ is always greater than equal to $-d$, where $\ringsz^0, A \in \bbZ$ are as defined in \theoremref{general-division}.
Thus, the third comparison ($A < -d$) in the expression of $C$ from \theoremref{general-division} can be omitted.
Moreover, this reduces the value of $\delta$ needed and  $\delta = \lceil \log 4d \rceil$ suffices since $-2d \leq A - d, A \leq 2d - 2$.

Our homomorphic encryption scheme requires $\ringsz$ to be a prime of the form $2 K \polymod + 1$ (\sectionref{he}), where $K$ is a positive integer and $\polymod\geq8192$ is a power-of-$2$.
Thus, we have $\ringsz^0 = \ringsz \bmod{2^s} = 1$ for $1 \leq s \leq 14$.
For all our benchmarks, $s \leq 12$ and we use this optimization for truncation in $\toolhe$. \\

 \section{Secure Inference}
\label{sec:secureinference}
We give an overview of all the layers  that must be computed securely to realize the task of secure neural network inference. Layers can be broken into two categories - {\em linear} and {\em non-linear}. An inference algorithm simply consists of a sequence of layers of appropriate dimension connected to each other. 
Examples of linear layers include matrix multiplication, convolutions, $\avgpool$ and batch normalization, while non-linear layers include $\relu$, $\maxpool$, and $\argmax$.

We are in the setting of secure inference where the model owner, say \alice, holds the weights. When securely realizing each of these layers, we maintain the following invariant: Parties \alice and \bob begin with arithmetic shares of the input to the layer and after the protocol, end with arithmetic shares (over the same ring) of the output of the layer. This allows us to stitch protocols for arbitrary layers sequentially to obtain a secure computation protocol for any neural network comprising of these layers. 
Semi-honest security of the protocol will follow trivially from sequential composibility of individual sub-protocols~\cite{gmw,canetti00,lindellsim}.
For protocols in $\toolot$, this arithmetic secret sharing is over $\bbZ_\intx$; in $\toolhe$,  the sharing is over $\bbZ_\ringsz$, prime $\ringsz$.
The inputs to secure inference are floating-point numbers, encoded as fixed-point integers in the ring ($\bbZ_{\intx}$ or $\bbZ_{\ringsz}$); for details see  \appendixref{input-encoding}.

\subsection{Linear Layers}

\subsubsection{Fully connected layers and convolutions.}
\label{sec:fcc}
 A fully connected layer in a neural network is simply a product of two matrices - the matrix of weights and the matrix of activations of that layer - of appropriate dimension. At a very high level, a convolutional layer applies a filter (usually of dimension $f\times f$ for small integer $f$) to the input matrix by sliding across it and computing the sum of element-wise products of the filter with the input. Various parameters are associated with convolutions - e.g. stride (a stride of 1 denotes that the filter slides across the larger input matrix beginning at every row and every column) and zero-padding (which indicates whether the matrix is padded with 0s to increase its dimension before applying the filter). 
When performing matrix multiplication or convolutions over fixed-point values, the values of the final matrix must be scaled down appropriately so that it has the same scale as the inputs to the computation. 
Hence, to do faithful fixed-point arithmetic, we first compute the matrix multiplication or convolution over the ring ($\bbZ_\intx$ or $\bbZ_\ringsz$) followed by truncation, i.e., division-by-$2^s$ of all the values.
In $\toolot$, multiplication and convolutions over the ring $\bbZ_\intx$ are done using oblivious transfer techniques and in $\toolhe$ these are done over $\bbZ_\ringsz$ using homomorphic encryption techniques
that we describe next followed by our truncation method.

\newcommand{\done}{M}
\newcommand{\dtwo}{N}
\newcommand{\dthree}{K}

\paragraph{OT based computation.}
The OT-based techniques for multiplication are well-known~\cite{beaver,aby,secureml} and we describe them briefly for completeness.
First consider the simple case of secure multiplication of $a$ and $b$ in $\bbZ_\intx$ where \alice knows $a$ and \alice and \bob hold arithmetic shares of $b$. 
This can be done by invoking $\iknpcot{i} \text{ for } i \in \{1, \ldots, \ell\}$ requiring communication equivalent to $\ell$ instances of $\iknpcot{\frac{\ell+1}{2}}$.
Using this, multiplying two matrices $A \in \bbZ_\intx^{\done,\dtwo}$ and $B \in \bbZ_\intx^{\dtwo,\dthree}$ such that \alice knows $A$ and $B$ is arithmetically secret shared requires $\done\dtwo\dthree\ell$ instances of $\iknpcot{\frac{\ell+1}{2}}$. 
This can be optimized with structured multiplications inside a matrix multiplication by combining all the COT sender messages when multiplying with the same element, reducing the complexity to that of $\dtwo\dthree\ell$ instances of $\iknpcot{\frac{\done(\ell+1)}{2}}$. 
Finally, we reduce the task of secure convolutions to secure matrix multiplication similar to \cite{aby3,securenn,cryptflow}.

\paragraph{HE based computation.} $\toolhe$ uses techniques from Gazelle~\cite{gazelle} and Delphi~\cite{delphi} to compute matrix multiplications and convolutions over a field $\bbZ_\ringsz$ (prime $\ringsz$), of appropriate size.
At a high level, first, \bob sends an encryption of its arithmetic share to \alice. Then, \alice homomorphically computes on this ciphertext using weights of the model (known to \alice) to compute an encryption of the arithmetic share of the result and sends this back to \bob.
Hence, the communication only depends on the input and output size of the linear layer and is independent of the number of multiplications being performed.
Homomorphic operations can have significantly high computational cost - to mitigate this, we build upon the {\em output rotations} method from~\cite{gazelle} for performing convolutions, and reduce its number of homomorphic rotations.
At a very high level, after performing convolutions homomorphically, ciphertexts are grouped, rotated in order to be correctly aligned, and then packed using addition.
In our work, we divide the groups further into subgroups that are misaligned {\em by the same offset}.
Hence the ciphertexts within a subgroup can first be added and the resulting ciphertext can then be aligned using a single rotation as opposed to subgroup-size many rotations in~\cite{gazelle}.
We refer the reader to \appendixref{opti-gazelle-algo} for details.

\paragraph{Faithful truncation.} To correctly emulate fixed-point arithmetic, the value encoded in the shares obtained from the above methods needs to be divided-by-$2^s$, where $s$ is the scale used. 
For this we invoke $\ftruncint{\ell}{s}$ in $\toolot$ and $\fdivring{\ringsz}{2^s}$ in $\toolhe$ for each value of the resulting matrix. With this, result of secure implementation of fixed-point multiplication and convolutions is {\em bitwise equal} to the corresponding cleartext execution. In contrast, many prior works on 2PC~\cite{secureml,delphi} and 3PC~\cite{securenn,aby3,cryptflow} used a {\em local truncation} method for approximate truncation based on a result from~\cite{secureml}. Here, the result can be arbitrarily wrong with a (small) probability $p$ and with probability $1-p$ the result can be wrong in the last bit. Since $p$ grows with the number of truncations,  these probabilistic errors are problematic for large DNNs. 
Moreover, even if $p$ is small, $1$-bit errors can accumulate and the results of cleartext execution and secure execution can diverge; this  is undesirable as it breaks correctness of 2PC.

\subsubsection{$\avgpool_d$} The function $\avgpool_d(a_1, \cdots, a_d)$ over a pool of $d$ elements $a_1, \cdots, a_d$ is defined to be the arithmetic mean of these $d$ values. The protocol to compute this function works as follows: $\alice$ and $\bob$ begin with arithmetic shares (e.g. over $\bbZ_\intx$ in $\toolot$) of $a_i$, for all $i \in [d]$. They perform local addition to obtain shares of $w = \sum_{i=1}^d a_i$ (i.e., $\party{b}$ computes $\share{w}{\intx}{b} = \sum_{i=1}^d \share{a_{i}}{\intx}{b}$). Then, parties invoke $\fdivring{L}{d}$ on inputs $\share{w}{\intx}{b}$ to obtain the desired output. Correctness and security follow in the $\fdivring{L}{d}-$hybrid model. Here too, unlike~\cite{delphi}, our secure execution of average pool is bitwise equal to the cleartext version.

\subsection{Nonlinear Layers}
\label{sec:non-linear}

\subsubsection{$\relu$}
\label{sec:prot-relu-int}
Note that $\relu(a) = a$ if $a \geq 0$, and $0$ otherwise. Equivalently, $\relu(a) = \drelu(a) \cdot a$.
For $\bbZ_\intx$, first we compute the boolean shares of $\drelu(a)$ using a call to $\fdreluint{\ell}$ and then we compute shares of $\relu(a)$ using a call to multiplexer  $\fmux{\intx}$ (\sectionref{mux}). We describe the protocol for $\relu(a)$ over $\bbZ_\intx$ formally in \algoref{relu-int}, \appendixref{relu-protocol} (the case of $\bbZ_\ringsz$ follows in a similar manner). For communication complexity, refer to \tableref{comp-relu} for comparison with garbled circuits and \appendixref{relu-protocol} for details.

\subsubsection{$\maxpool_d$ and $\argmax_d$} The function $\maxpool_d(a_1, \cdots, a_d)$ over $d$ elements $a_1, \cdots, a_d$ is defined in the following way. Define $\gt(x,y) = z$, where $w = x-y$ and $z = x$, if $w>0$ and $z=y$, if $w\leq 0$. Define $z_1 = a_1$ and $z_i = \gt(a_i,z_{i-1})$, recursively for all $2 \leq i \leq d$. Now, $\maxpool_d(a_1, \cdots, a_d) = z_d$.

We now describe a protocol such that parties begin with arithmetic shares (over $\bbZ_\intx$) of $a_i$, for all $i \in [d]$ and end the protocol with arithmetic shares (over $\bbZ_\intx$) of $\maxpool_d(a_1,\cdots,a_d)$. For simplicity, we describe how $\alice$ and $\bob$ can compute shares of $z = \gt(x,y)$ (beginning with the shares of $x$ and $y$). It is easy to see then how they can compute $\maxpool_d$. First, parties locally compute shares of $w = x-y$ (i.e., $\party{b}$ computes $\share{w}{\intx}{b} = \share{x}{\intx}{b}-\share{y}{\intx}{b}$, for $b\in \zo$). Next, they invoke $\fdreluint{\ell}$ with input $\share{w}{\intx}{b}$  to learn output $\share{v}{B}{b}$. Now, they invoke $\fmux{\intx}$ with input $\share{w}{\intx}{b}$ and $\share{v}{B}{b}$ to learn output $\share{t}{\intx}{b}$. Finally, parties output $\share{z}{\intx}{b} = \share{y}{\intx}{b}+\share{t}{\intx}{b}$. The correctness and security of the protocol follows in a straightforward manner. Computing $\maxpool_d$ is done using $d-1$ invocations of the above sub-protocol in $d-1$ sequential steps. 

$\argmax_d(a_1,\cdots,a_d)$ is defined similar to $\maxpool_d(a_1, \cdots, a_d)$, except that its output is an index $i^*$ s.t. $a_{i^*} = \maxpool_d(a_1,\cdots,a_d)$. $\argmax_d$ can be computed securely similar to $\maxpool_d(a_1, \cdots, a_d)$.

 \section{Implementation}
\label{sec:impl}

We implement our cryptographic protocols in a library and integrate them
into the CrypTFlow framework~\cite{cryptflow,cryptflowimpl} as a new cryptographic
backend. CrypTFlow compiles high-level TensorFlow~\cite{tensorflow} inference code to
secure computation protocols using its frontend Athos, that are then executed by its
cryptographic backends. We modify the truncation behavior of
 Athos in support of faithful
fixed-point arithmetic.
We start by describing the implementation of our cryptographic library, followed by the modifications that we made to Athos.

\subsection{Cryptographic backend}\label{subsec:cryptoimplementations}
To implement our protocols, we build upon the $\kkot{2}{\ell}$
implementation from EMP \cite{emp-toolkit} and extend it to
$\kkot{k}{\ell}$ using the protocol from \cite{kkot}. 
Our linear-layer
implementation in $\toolhe$ is based on
SEAL/Delphi~\cite{sealcrypto,github-delphi} and in $\toolot$ is based on EMP.
All our protocol implementations are multi-threaded.

\paragraph{Oblivious Transfer.} 

$\kkot{k}{\ell}$ requires a correlation robust function to mask the sender's messages in the OT extension protocol, and we use $\mathsf{AES}_{256}^{\mathsf{RK}}$ (re-keyed $\mathsf{AES}$ with $256$-bit key)\footnote{There are two types of $\mathsf{AES}$ in MPC applications - fixed-key~(FK) and re-keyed~(RK)~\cite{BHKR13,GKWY20}.  While the former runs key schedule only once and is more efficient, the latter generates a new key schedule for every invocation and is required in this application.} to instantiate it as in \cite{DKSSZZ17,aby}.
We incorporated the optimizations from
\cite{fast-std-garbling,aes-gcm-siv} for $\mathsf{AES}$ key expansion
and pipelining these $\mathsf{AES}_{256}^{\mathsf{RK}}$ calls. This leads to roughly $6\times$
improvement in the performance of $\mathsf{AES}_{256}^{\mathsf{RK}}$ calls,
considerably improving the overall execution time of $\kkot{k}{\ell}$
(e.g. $2.7\times$ over LAN for $\kkot{16}{8}$).

\paragraph{Millionaires' protocol.} Recall that $\blsize$ is a parameter in our protocol 
$\protmill{\ell,\blsize}$. While we discussed the dependence of communication complexity on $\blsize$ in \sectionref{comm-mill}, here we discuss its influence on the computational cost.
Our protocol makes $\ell/\blsize$ calls to $\kkot{\twoblsize}{2}$ (after merging steps~\ref{mill-ot1}\&\ref{mill-ot2}), where $\twoblsize = 2^\blsize$. 
Using OT extension techniques, generating an instance of $\kkot{\twoblsize}{2}$ requires $6~\mathsf{AES}_{256}^{\mathsf{FK}}$ and $(\twoblsize+1)~\mathsf{AES}_{256}^{\mathsf{RK}}$ evaluations. Thus, the computational cost grows super-polynomially with $\blsize$.
We note that for $\ell=32$, even though communication is minimized for $\blsize=7$, empirically we observe that $\blsize = 4$ gives us the best performance under both LAN and WAN settings (communication in this case is about $30\%$ more than when $\blsize=7$ but computation is $\approx 3\times$ lower).

\paragraph{Implementing linear layers in $\toolhe$.}
To implement the linear layers in $\toolhe$, we build upon the Delphi
implementation \cite{github-delphi,delphi}, that is in turn based on
the SEAL library~\cite{sealcrypto}.
We use the code for fully connected layers as it is from~\cite{github-delphi}.
For convolution layers, we parallelize the code, employ modulus-switching~\cite{sealcrypto} to reduce the ciphertext modulus (and hence ciphertext size), and implement the strided convolutions proposed in Gazelle~\cite{gazelle}.
These optimizations resulted in significant performance improvement of
convolution layers. E.g. for the first convolution
layer\footnote{Layer parameters: image size $230 \times 230$, filter
  size $7 \times 7$, input channels $3$, output channels $64$, and
  stride size $2 \times 2$} of \resnet, the runtime decreased from
$306$s to $18$s  in the LAN setting and
communication decreased from $204$ MiB to $76$ MiB.

\subsection{CrypTFlow integration}
\label{sec:compiler}

We integrate $\toolot$ and $\toolhe$ as new cryptographic backends into the CrypTFlow framework~\cite{cryptflow,cryptflowimpl}.
Thus, as in CrypTFlow~\cite{cryptflow}, we can work with unmodified TensorFlow code as input to produce our secure computation protocols.
CrypTFlow's TensorFlow frontend Athos outputs fixed-point DNNs that use
64-bit integers and sets an optimal scale using a validation set. CrypTFlow required a {\em bitwidth} of 64 to
ensure that the probability of local truncation errors in its protocols is small 
(Section~\ref{sec:fcc}). Since our protocols are correct and have no such errors, 
we extend Athos to set both the bitwidth and the scale optimally by autotuning on the validation set.
The bitwidth and scale leak information about the weights and this leakage is similar to the prior works on secure inference~\cite{secureml,delphi,gazelle,minionn,cryptflow,securenn,aby3}.

Implementing faithful  truncations using $\prottruncint{\ell}{s}$ requires the parties to communicate. We implement the following peephole optimizations in Athos to reduce the cost of
these truncation calls.
Consider a DNN having a convolution layer
followed by a ReLU layer. While truncation can be done
immediately after the convolution, moving the truncation call
to after the ReLU layer can reduce the cost of our protocol
$\prottruncint{\ell}{s}$. 
Since the values after ReLU are guaranteed to be all positive, the call to $\fdreluint{\ell}$ within it
(step 2 in \algoref{truncate-int}) now becomes redundant and can be omitted. 
Our optimization further accounts for operations that may occur
between the convolutions  and ReLU, say a matrix addition.
Moving the truncation call from immediately after convolution to after
ReLU means the activations flowing into the addition
operation are now scaled by $2s$, instead of the usual $s$. For the
addition operation to then work correctly, we scale the other argument
of addition by $s$ as well. These optimizations are fully automatic and need
no manual intervention.

 \section{Experiments}
\label{sec:experiments}

We empirically validate the following claims:
\begin{itemize}
\item In \sectionref{exp-comparison-gc}, we show that our protocols for computing $\relu$ activations are more efficient than state-of-the-art garbled circuits-based implementations (\tableref{exp-relu}). Additionally, our division protocols outperforms garbled circuits when computing average pool layers. 
\item  On the DNNs considered by prior work on secure inference, our protocols can evaluate the non-linear layers  much more efficiently and decrease the total time (\tableref{exp-comp-prior-work}) as well as the online time (\tableref{exp-delphi-online}).
\item We show the first empirical evaluation of 2-party secure inference on ImageNet-scale benchmarks (\sectionref{exp-eval-DNN}). These results show the trade-offs between OT and HE-based secure DNN inference (\tableref{eval-DNN}).
\end{itemize}
We start with a description of our experimental setup and benchmarks, followed by the results.

\paragraph{Experimental Setup.}
We ran our benchmarks in two network settings, namely, a LAN setting with both machines situated in West Europe, and a transatlantic WAN setting with one of the machines in East US.
The bandwidth between the machines is 377 MBps and 40 MBps in the LAN and the WAN setting respectively and  the echo latency is 0.3ms and 80ms respectively.
Each machine has commodity class hardware: 3.7 GHz Intel Xeon processor with 4 cores and 16 GBs of RAM.

\paragraph{Our Benchmarks.}
We evaluate on the ImageNet-scale benchmarks considered by~\cite{cryptflow}: \squeezenet~\cite{squeezenet}, \resnet~\cite{resnet}, and \densenet~\cite{densenet}. 
To match the reported accuracies, we need 37-bit fixed-point numbers for \resnet, whereas 32 bits suffice for \densenet and \squeezenet (\appendixref{accuracy-summary}). Recall that our division protocols lead to correct secure executions and there is no accuracy loss in going from cleartext inference to secure inference.~\appendixref{DNN-summary} provides a brief summary of these benchmarks.

\begin{figure}[t]
\centering
    \includegraphics[width=0.475\textwidth]{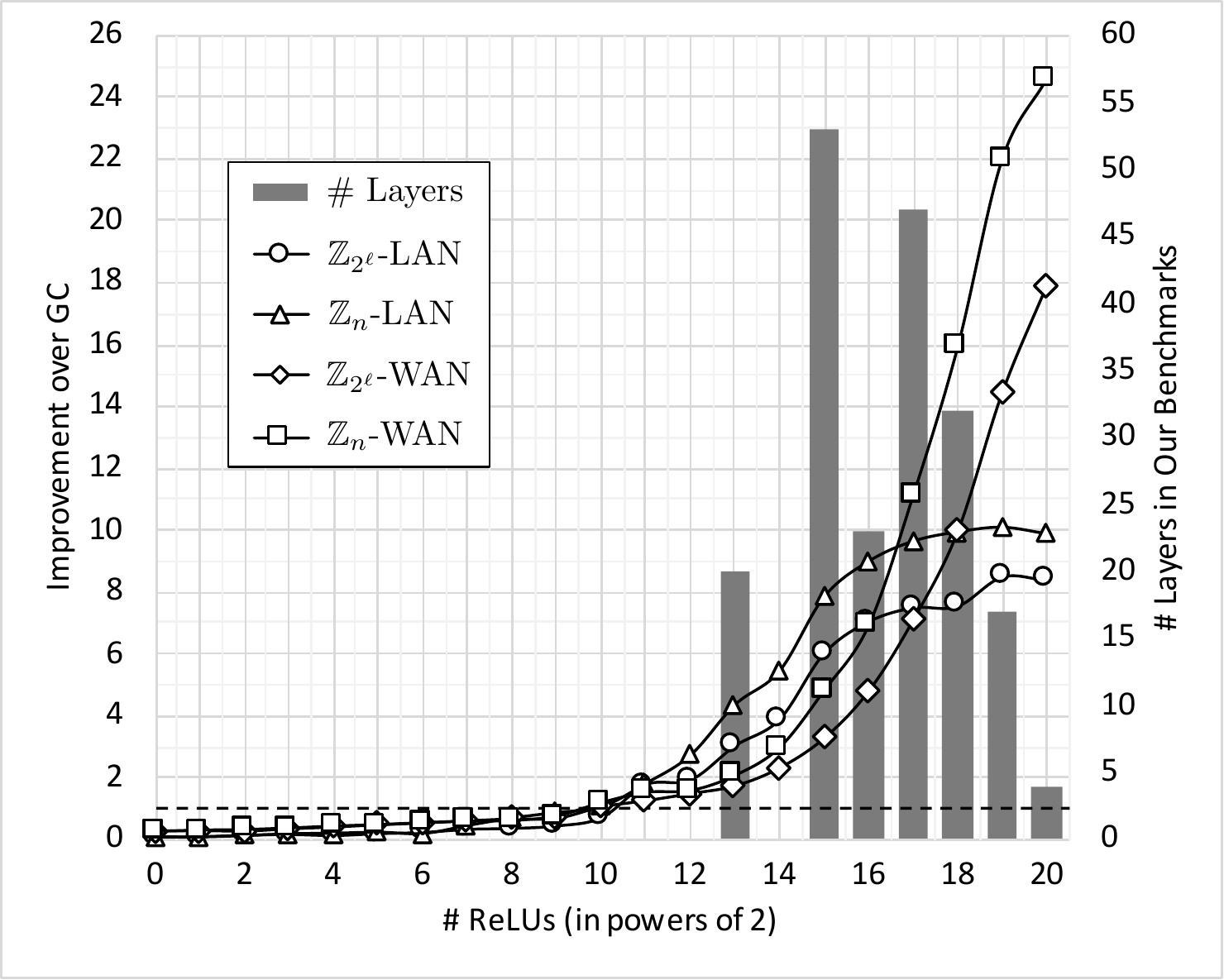}
    \caption{The left y-axis shows ($\frac{\text{GC Time}}{\text{Our Time}}$). The right y-axis shows the total number of $\relu$ layers corresponding to each layer size in our benchmark set. The legend entries denote the input domain and the network setting.}
    \label{fig:exp-relu-scale}
\end{figure}
 \begin{table}
    \centering
    \begin{subtable}{0.5\textwidth}
        \centering
        \begin{tabular}{|c|c|c|c|c|c|c|}
            \hline
            \multirow{2}{*}{Benchmark} & \multicolumn{3}{c|}{Garbled Circuits} & \multicolumn{3}{c|}{Our Protocols} \\
            & LAN & WAN & Comm & LAN & WAN & Comm \\ \hline \hline
            \squeezenet & 26.4 & 265.6 & 7.63 & 3.5 & 33.3 & 1.15 \\ \hline
            \resnet & 136.5 & 1285.2 & 39.19 & 16.4 & 69.4 & 5.23 \\ \hline
            \densenet & 199.6 & 1849.3 & 56.57 & 24.8 & 118.7 & 8.21\\ \hline
        \end{tabular}
        \caption{over $\bbZ_{2^{\ell}}$} \label{subtab:exp-relu-ring}
    \end{subtable}\\\smallbreak
    \begin{subtable}{0.5\textwidth}
        \centering
        \begin{tabular}{|c|c|c|c|c|c|c|}
            \hline
            \multirow{2}{*}{Benchmark} & \multicolumn{3}{c|}{Garbled Circuits} & \multicolumn{3}{c|}{Our Protocols} \\
            & LAN & WAN & Comm & LAN & WAN & Comm \\ \hline \hline
            \squeezenet & 51.7 & 525.8 & 16.06 & 5.6 & 50.4 & 1.77 \\ \hline
            \resnet & 267.5 & 2589.7 & 84.02 & 28.0 & 124.0 & 8.55 \\ \hline
            \densenet & 383.5 & 3686.2 & 118.98 & 41.9 & 256.0 & 12.64 \\ \hline
        \end{tabular}
        \caption{over $\bbZ_{\ringsz}$} \label{subtab:exp-relu-field}
    \end{subtable}
    \vspace{-5pt}
    \caption{Performance comparison with Garbled Circuits for $\relu$ layers. Runtimes are in seconds and comm.  in GiB.} \label{tab:exp-relu}
\end{table}
 \subsection{Comparison with Garbled Circuits}
\label{sec:exp-comparison-gc}
We compare with EMP-toolkit \cite{emp-toolkit}, the state-of-the-art library for Garbled Circuits (GC).
\figureref{exp-relu-scale} shows the improvement of our $\relu$ protocols over GC in both LAN and WAN settings.
On the x-axis, which is in log-scale, the number of $\relu$s range from $2^0$ to $2^{20}$.
The histogram shows, using the right y-axis, the cumulative number of layers in our benchmarks (\squeezenet, \resnet, \densenet) which require the number of $\relu$ activations given on the x-axis. We observe that these DNNs have layers that compute between $2^{13}$ and $2^{20}$ $\relu$s. For such layers, we observe (on the left y-axis) that our protocols are $2\times$--$25\times$ faster than GC -- the larger the layers the higher the speedups, and gains are larger in the WAN settings.
Specifically, for WAN and $> 2^{17}$ $\relu$s, the speedups are much higher than the LAN setting. Here,
the cost of rounds is amortized over large layers and the communication cost is a large fraction of the total runtime.
Note that our implementations perform load-balancing to leverage full-duplex  TCP.

Next, we compare the time taken by GC and our protocols in computing the $\relu$ activations of our benchmarks in \tableref{exp-relu}.
Our protocol over $\bbZ_{\intx}$ is up to $8 \times$ and $18 \times$ faster than GC in the LAN and WAN settings respectively, while it is $\approx 7\times$ more communication efficient.
As expected, our protocol over $\bbZ_{\ringsz}$ has even better gains over GC.
Specifically, it is up to $9 \times$ and $21\times$ faster in the LAN and WAN settings respectively, and has $\approx 9\times$ less communication.

We also performed a similar comparison of our protocols with GC for the $\avgpool$ layers of our benchmarks, and saw up to $51\times$ reduction in runtime and $41\times$ reduction in communication.
We report the concrete performance numbers and discuss the results in more detail in \appendixref{avgpool-numbers}.

\subsection{Comparison with Delphi}
\label{sec:exp-comparison-prior-work}
In this section, we compare with Delphi~\cite{delphi}, which is the current state-of-the-art for $2$-party secure DNN inference that outperforms~\cite{gazelle,ezpc,chameleon,minionn,cryptonets,nhe1,nhe2,chet,hycc}
in total time as well as the time taken in online phase. 
It uses garbled circuits for non-linear layers, and we show that with our protocols, the time taken to evaluate the non-linear layers  can be decreased significantly.

For a fair evaluation, we demonstrate these improvements on the benchmarks of Delphi~\cite{delphi}, i.e., the \textsf{MiniONN} (CIFAR-10)~\cite{minionn} and \textsf{ResNet32} (CIFAR-100) DNNs with $\relu$ activations (as opposed to the ImageNet-scale benchmarks for which Delphi has not been optimized). Similar to Delphi, we perform these computations with a bitwidth of $41$ in the LAN setting.

\begin{table}
    \centering
    \begin{tabular}{|c|c|c|c|c|c|}
        \hline
        \multirow{2}{*}{Benchmark} &  \multirow{2}{*}{Metric} & \multirow{2}{*}{Linear} & \multicolumn{3}{c|}{Non-linear} \\
        & & & Delphi & Ours & Improvement \\ \hline \hline
        \multirow{2}{*}{\textsf{MiniONN}} & Time & 10.7 & 30.2 & 1.0 & $30.2\times$ \\  & Comm. & 0.02 & 3.15 & 0.28 & $12.3\times$\\ \hline
        \multirow{2}{*}{\textsf{ResNet32}} & Time & 15.9 & 52.9 & 2.4 & $22.0\times$ \\  & Comm. & 0.07  & 5.51 & 0.59 & $9.3\times$\\ \hline
    \end{tabular}
    \caption{Performance comparison  with Delphi~\cite{delphi} for non-linear layers. Runtimes are in seconds and comm. in GiB.} \label{tab:exp-comp-prior-work}
\end{table}
 \begin{table}
    \centering
    \begin{tabular}{|c|c|c|c|c|}
        \hline
        \multirow{2}{*}{Benchmark} & \multirow{2}{*}{Linear} & \multicolumn{3}{c|}{Non-linear} \\
        & & Delphi & Ours & Improvement \\ \hline \hline
        \multirow{1}{*}{\textsf{MiniONN}} & < 0.1 & 3.97 & 0.32 & $12.40\times$ \\ \hline \multirow{1}{*}{\textsf{ResNet32}} & < 0.1 & 6.99 & 0.63 & $11.09\times$ \\ \hline \end{tabular}
    \caption{Performance comparison with Delphi~\cite{delphi} for online runtime in seconds.} \label{tab:exp-delphi-online}
\end{table}
 
In \tableref{exp-comp-prior-work}, we report the performance of Delphi for evaluating the linear and non-linear components of \textsf{MiniONN} and \textsf{ResNet32} separately, along with the performance of our protocols for the same non-linear computation\footnote{Our non-linear time includes the cost of correct truncation.}. The table shows that the time to evaluate non-linear layers is the bulk of the total time and our protocols
are $20\times$--$30\times$ faster in evaluating the non-linear layers.
Also note that we reduce the communication by $12\times$ on \textsf{MiniONN}, and require  {$9\times$} less communication on \textsf{ResNet32}.

Next, we compare the online time of our protocols with the online time of Delphi in \tableref{exp-delphi-online}. In the online phase, linear layers take negligible time and all the time is spent in evaluating the non-linear layers. Here, our protocols are an order of magnitude more efficient than Delphi.

\subsection{Evaluation on practical DNNs}
\label{sec:exp-eval-DNN}
With all our protocols and implementation optimizations in place, we demonstrate the scalability of \tool by efficiently running ImageNet-scale secure inference.
Table~\ref{tab:eval-DNN} shows that both our backends, $\toolot$ and $\toolhe$, are efficient enough to evaluate $\squeezenet$ in under a minute and scale to $\resnet$ and $\densenet$. 

In the LAN setting, for both $\squeezenet$ and $\densenet$, $\toolot$ performs better than $\toolhe$ by at least $20\%$ owing to the higher compute in the latter. 
However, the quadratic growth of communication with bitlength in the linear-layers of $\toolot$ can easily drown this difference if we go to higher bitlengths. Because $\resnet$, requires 37-bits (compared to 32 in $\squeezenet$ and $\densenet$) to preserve accuracy, $\toolhe$ outperforms $\toolot$ in both LAN and WAN settings. 
In general for WAN settings where communication becomes the major performance bottleneck, $\toolhe$ performs better than $\toolot$: $2\times$ for $\squeezenet$ and $\densenet$ and $4\times$ for $\resnet$.
Overall, with $\tool$, we could evaluate all the 3 benchmarks within 10 minutes on LAN and 20 minutes on WAN. 
Since \tool supports both $\toolot$ and $\toolhe$, one can choose a specific backend depending on the network statistics~\cite{hycc,cheapsmc} to get the best secure inference latency. To the best of our knowledge, no prior system provides this support for OT and HE-based secure DNN inference.

\begin{table}
    \centering
    \begin{tabular}{|c|c|c|c|c|}
        \hline
        Benchmark & Protocol & LAN & WAN & Comm \\ \hline \hline
        \multirow{2}{*}{\squeezenet} & $\toolot$ & 44.3 & 293.6 & 26.07 \\ \cline{2-5}
        & $\toolhe$ & 59.2 & 156.6 & 5.27 \\ \hline
        \multirow{2}{*}{\resnet} & $\toolot$ & 619.4 & 3611.6 & 370.84 \\ \cline{2-5}
        & $\toolhe$ & 545.8 & 936.0 & 32.43 \\ \hline
        \multirow{2}{*}{\densenet} & $\toolot$ & 371.4 & 2257.7 & 217.19 \\ \cline{2-5}
        & $\toolhe$ & 463.2 & 1124.7 & 35.56 \\ \hline
    \end{tabular}
    \caption{Performance of \tool on ImageNet-scale benchmarks. Runtimes are in seconds and comm. in GiB.} \label{tab:eval-DNN}
\end{table}
  \section{Conclusion and Future Work}
\label{sec:conclusion}
We have presented secure, efficient, and correct implementations of practical 2-party DNN  inference that outperform prior work~\cite{delphi} by an order of magnitude in both latency and scale.
We evaluate the first secure implementations of ImageNet scale inference, a task that previously
required 3PC protocols~\cite{cryptflow,quantizednn} (which provide weaker security guarantees) or leaking intermediate computations~\cite{nhe2}.
In the future, we would like to consider ImageNet scale secure {\em training}. 
Even though we can run inference on commodity machines, for training
we would need protocols that can leverage specialized compute and networking hardware.
Like all prior work on 2PC for secure DNN inference, \tool only considers semi-honest adversaries. 
In the future, we would like to consider malicious adversaries.
Another future direction is to help the server in hiding $F$ from the client when computing a classifier $F(x,w)$. Like~\cite{gazelle}, $\toolhe$ can hide some aspects of  $F$: the filter sizes, the strides, and whether a layer is convolutional or fully connected. 
Thus, $\toolhe$ hides more information than OT-based tools~\cite{minionn} but reveals more information than FHE-based tools~\cite{nhe1,cryptonets}. 
We are exploring approaches to hide more information about $F$ while incurring minimal overhead.
 
\bibliographystyle{ACM-Reference-Format}
\bibliography{main}

\appendix
\section{Supporting Protocols}
\label{app"supp-protocols}

Here, we describe supporting protocols that our main protocols rely on.

\subsection{Protocol for regular $\fand$}
\label{app:regular-bit-triple}

Regular $\fand$ can be realized using bit-triples~\cite{beaver}, which are of the form $(\share{d}{B}{b}, \share{e}{B}{b}, \share{f}{B}{b})$, where $b \in \zo$ and $d \wedge e = f$.
Using an instance of $\kkot{16}{2}$, the parties can generate two bit-triples~\cite{DKSSZZ17}.
We describe this protocol for generating the first triple, and from there, it will be easy to see how to also get the second triple using the same OT instance.
The parties start by sampling random shares $\share{d}{B}{b}, \share{e}{B}{b} \getsr \zo$ for $b \in \zo$.
\bob sets the first two bits of its input to $\kkot{16}{2}$ as $\share{d}{B}{1} \concat \share{e}{B}{1}$, while the other two bits are used for the second triple.
\alice samples a random bit $r$ and sets its input messages to $\kkot{16}{2}$ as follows: for the $i$-th message, where $i \in \zo^4$, \alice uses the first two bits $i_1 \concat i_2$ of $i$ to compute $r \xor ((i_1 \xor \share{d}{B}{0}) \wedge (i_2 \xor \share{e}{B}{0}))$, and sets it as the first bit of the message, while reserving the second bit for the other triple.
Finally, \alice sets $\share{f}{B}{0} = r$, and \bob sets the first bit of the output of $\kkot{16}{2}$ as $\share{f}{B}{1}$.
It is easy to see correctness by noting that $\share{f}{B}{1} = \share{f}{B}{0} \xor (d \wedge e)$, and since $\share{f}{B}{0}$ is uniformly random, security follows directly in the $\kkot{16}{2}$-hybrid.

The communication of this protocol is the same as that of $\kkot{16}{2}$, which is $2 \secpar + 16 \cdot 2$ bits.
Since we generate two bit-triples using this protocol, the amortized cost per triple is $\secpar + 16$ bits, which is $144$ for $\secpar = 128$.

\subsection{Protocol for correlated $\fand$}
\label{app:correlated-bit-triple}
Correlated triples are two sets of bit triples $(\share{d}{B}{b}$, $\share{e}{B}{b}$, $\share{f}{B}{b})$ and $(\share{{d'}}{B}{b}$, $\share{e'}{B}{b}$, $\share{{f'}}{B}{b})$, for $b \in \zo$, such that $e = e'$, $d \wedge e = f$, and $d' \wedge e' = f'$.
The protocol from \appendixref{regular-bit-triple} required a $\kkot{16}{2}$ invocation to generate two regular triples, where the $4$ bits of \bob's input were its shares of $d, e, d',$ and $e'$.
However, when generating correlated triples, we can instead use an instance of $\kkot{8}{2}$ because $e = e'$, and thus, 3 bits suffice to represent \bob's input.
Correctness and security follow in a similar way as in the case of regular $\fand$ (see \appendixref{regular-bit-triple}).

The communication of this protocol is equal to that of $\kkot{8}{2}$, which costs $2 \secpar + 8 \cdot 2$ bits.
Thus, we get an amortized communication of $\secpar + 8$ bits per correlated triple.

\subsection{Protocol for Multiplexer}
\label{app:mux}
We describe our protocol for realizing $\fmux{\ringsz}$ in \algoref{mux}.

\begin{algorithm}[t]
\caption{Multiplexer,  $\protmux{\ringsz}$:}
\label{algo:mux}
\begin{algorithmic}[1]
\Require For $b \in \zo$,  $\party{b}$ holds $\share{a}{\ringsz}{b}$ and $\share{c}{B}{b}$.
\Ensure For $b \in \zo$, $\party{b}$ learns $\share{z}{\ringsz}{b}$ s.t. $z = a$ if $c = 1$, else $z = 0$.

\vspace{0.2cm}

\State For $b \in \zo$, $\party{b}$ picks $r_b \getsr \bbZ_{\ringsz}$. 

\State \alice sets $s_0, s_1$ as follows: If $\share{c}{B}{0} = 0$, $(s_0, s_1) = (-r_0, -r_0+\share{a}{\ringsz}{0})$. Else, $(s_0, s_1) = (-r_0+\share{a}{\ringsz}{0}, -r_0)$.

\State \alice \& \bob invoke an instance of $\kkot{2}{\etax}$ where \alice is the sender with inputs $(s_0, s_1)$ and \bob is the receiver with input $\share{c}{B}{1}$. Let \bob's output be $x_1$.

\State \bob sets $t_0, t_1$ as follows: If $\share{c}{B}{1} = 0$, $(t_0, t_1) = (-r_1, -r_1+\share{a}{\ringsz}{1})$. Else, $(t_0, t_1) = (-r_1+\share{a}{\ringsz}{1}, -r_1)$.

\State \alice \& \bob invoke an instance of $\kkot{2}{\etax}$ where \bob is the sender with inputs $(t_0, t_1)$ and \alice is the receiver with input $\share{c}{B}{0}$. Let \alice's output be $x_0$.

\State For $b \in \zo$, $\party{b}$ outputs $\share{z}{\ringsz}{b} = r_b + x_b$.

\end{algorithmic}
\end{algorithm}
 
First we argue correctness.  Let $c =\reconst{B}{\share{c}{B}{0}, \share{c}{B}{1}} = \share{c}{B}{0} \xor \share{c}{B}{1}$. By correctness of $\kkot{2}{\etax}$, $x_1 = -r_0 + c\cdot \share{a}{\ringsz}{0}$. Similarly, $x_0 = -r_1 + c \cdot \share{a}{\ringsz}{1}$. Hence, $\reconst{\ringsz}{\share{z}{\ringsz}{0}, \share{z}{\ringsz}{1}} = z_0 + z_1 = c\cdot a$.
Security trivially follows in $\kkot{2}{\etax}$-hybrid. Communication complexity is $2(\secpar+2\etax)$.

\subsection{Protocol for B2A}
\label{app:BtoA}
We describe our protocol for realizing $\fBtoA{\ringsz}$ formally in \algoref{B2A}. For correctness, we need to show that $d = \reconst{\intx}{\share{d}{\ringsz}{0}, \share{d}{\ringsz}{1}}$ $= \share{c}{B}{0} + \share{c}{B}{1} - 2\share{c}{B}{0}\share{c}{B}{1}$. By correctness of $\iknpcot{\etax}$, $y_1 = x + \share{c}{B}{0}\share{c}{B}{1}$. Using this, $\share{d}{\ringsz}{0} = \share{c}{B}{0} + 2x$ and $\share{d}{\ringsz}{1} = \share{c}{B}{1} - 2x - 2\share{c}{B}{0}\share{c}{B}{1}$. Security follows from the security of $\iknpcot{\etax}$ and communication required is $\secpar + \etax$ bits.

\begin{algorithm}[t]
\caption{Boolean to Arithmetic, $\protbtoa{\ringsz}$:}
\label{algo:B2A}
\begin{algorithmic}[1]

\Require $\party{0}, \party{1}$ hold $\share{c}{B}{0}$ and $\share{c}{B}{1}$, respectively, where $c \in \zo$.
\Ensure $\party{0}, \party{1}$ learn $\share{d}{\ringsz}{0}$ and $\share{d}{\ringsz}{1}$, respectively, s.t. $d=c$.

\vspace{0.2cm}

\State \alice \& \bob invoke an instance of $\iknpcot{\etax}$ where \alice is the sender with correlation function $f(x)=x+\share{c}{B}{0}$ and \bob is the receiver with input $\share{c}{B}{1}$. Party \alice learns $x$ and sets $y_0 = \ringsz - x$ and \bob learns $y_1$. 

\State For $b \in \zo$, $\party{b}$ computes $\share{d}{\ringsz}{b}=\share{c}{B}{b}-2 \cdot y_b$. 

\end{algorithmic}
\end{algorithm}
 
  \begin{algorithm}[t]
\caption{$\ell$-bit  integer $\relu$, $\protreluint{\ell}$:}
\label{algo:relu-int}
\begin{algorithmic}[1]

 \Require $\party{0}, \party{1}$ hold $\share{a}{\intx}{0}$ and $\share{a}{\intx}{1}$, respectively. 
\Ensure $\party{0}, \party{1}$ get $\share{\relu(a)}{\intx}{0}$ and $\share{\relu(a)}{\intx}{1}$.

\vspace{0.2cm}
\State For $b\in \zo$, $\party{b}$ invokes $\fdreluint{\ell}$ with input $\share{a}{\intx}{b}$  to learn output $\share{y}{B}{b}$.

\State For $b\in \zo$, $\party{b}$ invokes $\fmux{\intx}$ with inputs $\share{a}{\intx}{b}$ and $\share{y}{B}{b}$ to learn $\share{z}{\intx}{b}$ and sets $\share{\relu(a)}{\intx}{b} = \share{z}{\intx}{b}$.

\end{algorithmic}
\end{algorithm}
 
\section{Protocol for $\relu$}
\label{app:relu-protocol}
We describe our $\relu$ protocol for the case where the input and output shares are over $\bbZ_{\intx}$ in \algoref{relu-int}, and note that the case of $\bbZ_{\ringsz}$ follows similarly.
It is easy to see that the correctness and security of the protocol follow in the $(\fdreluint{\ell}, \fmux{\intx})-$hybrid. \\

\noindent{\em Communication complexity.}
We first look at the complexity of $\protreluint{\ell}$, which involves a call to $\fdreluint{\ell}$ and $\fmux{\intx}$.
$\fdreluint{\ell}$ has the same communication as $\fmill{\ell-1}$, which requires $\secpar (\ell-1) + 13\frac{1}{2}(\ell-1) - 2\secpar - 22$ bits if we assume $\blsize=4$ and $\blsize \mid (\ell-1)$, and exclude optimization~(\ref{optimill:eq}) in the general expression from \sectionref{comm-mill}.
$\fmux{\intx}$ incurs a cost of $2\secpar + 4\ell$ bits, bringing the total cost to $\secpar \ell + 17\frac{1}{2}\ell - \secpar - 35\frac{1}{2}$ bits, which can be rewritten as $< \secpar \ell + 18 \ell$.
We get our best communication for $\ell=32$ (with all the optimizations) by taking $\blsize = 7$ for the $\protmill{31}$ invocation inside $\protdreluint{32}$, which gives us a total communication of $3298$ bits.

Now, we look at the complexity of $\protreluring{\ringsz}$, which makes calls to $\fdreluring{\ringsz}$ and $\fmux{\ringsz}$.
The cost of $\fdreluring{\ringsz}$ is $2 \secpar + 4$ bits for $\kkot{4}{1}$, plus $\frac{3}{2}\secpar(\etax+1)+27(\etax+1)-4\secpar-44$ bits for 2 invocations of $\fmill{\etax+1}$, where \bob's input is the same in both invocations and the same assumptions are made as for the expression of $\fmill{\ell-1}$ above.
The cost of $\fmux{\ringsz}$ is $2 \secpar + 4 \etax$ bits, and thus, the total cost is $\frac{3}{2}\secpar(\etax+1) + 31\etax - 13$, which can be rewritten as $<\frac{3}{2}\secpar(\etax+1)+31\etax$.
Concretely, we get the best communication for $\etax=32$ by taking $\blsize=7$ for the millionaire invocations, getting a total communication of $5288$ bits.
 \section{Proof of division theorem}
\label{app:division-proof}

Here, we prove \theoremref{general-division}.

\begin{proof}
    From \equationref{div-eqn}, we can write $\frdiv{\share{a}{\ringsz}{i}}{d}$ as:
\begin{align} \label{eq:share-div-correct}
    \frdiv{\share{a}{\ringsz}{i}}{d} &=_n{} \fidiv{a_i - \ind{a_i \geq \ringsz'} \cdot \ringsz}{d} \nonumber \\
          &=_{\ringsz}{} \fidiv{a_i^1 \cdot d + a_i^0 - \ind{a_i \geq \ringsz'} \cdot (\ringsz^1 \cdot d + \ringsz^0)}{d} \nonumber \\
          &=_{\ringsz}{} a_i^1 - \ind{a_i \geq \ringsz'} \cdot \ringsz^1 + \fidiv{a_i^0 - \ind{a_i \geq \ringsz'} \cdot \ringsz^0}{d},
\end{align}
for $i \in \zo$. $a_u$ can be expressed as $a_u = a_0 + a_1 - w \cdot \ringsz$, where the wrap-bit $w = \ind{a_0 + a_1 \geq \ringsz}$. We can rewrite this as:
\begin{align} \label{eq:secret-sum-of-shares-correct}
    a_u &= a_0 + a_1 - w \cdot \ringsz \nonumber \\
        &= (a_0^1 + a_1^1 - w \cdot \ringsz^1) \cdot d + (a_0^0 + a_1^0 - w \cdot \ringsz^0) \nonumber \\
        &= (a_0^1 + a_1^1 - w \cdot \ringsz^1 + k) \cdot d + (a_0^0 + a_1^0 - w \cdot \ringsz^0 - k \cdot d),
\end{align}
for some integer $k$ such that $0 \leq a_0^0 + a_1^0 - w \cdot \ringsz^0 - k \cdot d < d$.
    Similar to \equationref{share-div-correct} and from \equationref{secret-sum-of-shares-correct}, we can write $\frdiv{a}{d}$ as:
\begin{align} \label{eq:secret-div-correct}
    \frdiv{a}{d} =_{\ringsz}{} &a_0^1 + a_1^1 - w \cdot \ringsz^1 + k - \ind{a \geq \ringsz'} \cdot \ringsz^1 \nonumber \\
            &+ \fidiv{a_0^0 + a_1^0 - w \cdot \ringsz^0 - k \cdot d - \ind{a \geq \ringsz'} \cdot \ringsz^0}{d} \nonumber \\
        =_{\ringsz}{} &a_0^1 + a_1^1 - w \cdot \ringsz^1 - \ind{a \geq \ringsz'} \cdot \ringsz^1 \nonumber \\
            &+ \fidiv{a_0^0 + a_1^0 - w \cdot \ringsz^0 - \ind{a \geq \ringsz'} \cdot \ringsz^0}{d}.
\end{align}
From Equations \ref{eq:share-div-correct} and \ref{eq:secret-div-correct}, we have the following correction term:
\begin{align} \label{eq:div-error-correct}
    c =_{\ringsz}{} &\frdiv{a}{d} - \frdiv{\share{a}{\ringsz}{0}}{d} - \frdiv{\share{a}{\ringsz}{1}}{d} \nonumber \\
      =_{\ringsz}{} &\big(\ind{a_0 \geq \ringsz'} + \ind{a_1 \geq \ringsz'} - w - \ind{a \geq \ringsz'}\big) \cdot \ringsz^1 \nonumber \\
          &+ \fidiv{a_0^0 + a_1^0 - w \cdot \ringsz^0 - \ind{a \geq \ringsz'} \cdot \ringsz^0}{d} \\
          &- \big (\fidiv{a_0^0 - \ind{a_0 \geq \ringsz'} \cdot \ringsz^0}{d} + \fidiv{a_1^0 - \ind{a_1 \geq \ringsz'} \cdot \ringsz^0}{d} \big ) \nonumber \\
      =_{\ringsz}{} & c^1 \cdot \ringsz^1 + c^0 - B
\end{align}
Let $A_i' = \fidiv{a_0^0 + a_1^0 - i \cdot \ringsz^0}{d}$.
Then the values of the correction terms $c^1$ and $c^0$ are as summarized in \tableref{division-proof-corr-terms}.
\begin{table}
\centering
\begin{tabular}{|c|c|c|c|c|c|c|}
    \hline
    \# & $\ind{a_0 \geq \ringsz'}$ & $\ind{a_1 \geq \ringsz'}$ & $\ind{a_u \geq \ringsz'}$ & $w$ & $c^1$ & $c^0$ \\
    \hline
    \hline
    1 & 0 & 0 & 0 & 0 &  0 & $A_0'$ \\
    \hline
    2 & 0 & 0 & 1 & 0 & -1 & $A_1'$ \\
    \hline
    3 & 0 & 1 & 0 & 1 &  0 & $A_1'$ \\
    \hline
    4 & 0 & 1 & 1 & 0 &  0 & $A_1'$ \\
    \hline
    5 & 1 & 0 & 0 & 1 &  0 & $A_1'$ \\
    \hline
    6 & 1 & 0 & 1 & 0 &  0 & $A_1'$ \\
    \hline
    7 & 1 & 1 & 0 & 1 &  1 & $A_1'$ \\
    \hline
    8 & 1 & 1 & 1 & 1 &  0 & $A_2'$ \\
    \hline
\end{tabular}
\caption{Truth table for the correction terms $c^0$ and $c^1$ in the proof of division theorem (\appendixref{division-proof}).} \label{tab:division-proof-corr-terms}
\end{table}

From the table, we have $c^1 = \corr$ and can rewrite the correction term as $c =_{\ringsz} \corr \cdot \ringsz^1 + c^0 - B$.
Thus, adding $\corr \cdot \ringsz^1 - B \bmod{\ringsz}$ to $\frdiv{\share{a}{\ringsz}{0}}{d} + \frdiv{\share{a}{\ringsz}{1}}{d}$ accounts for all the correction terms except $c_0 \bmod{\ringsz}$.

Now all that remains to be proven is that $c^0 = 1 - C$.
Let $C_0~=~\ind{A~<~d}$, $C_1~=~\ind{A~<~0}$, and $C_2~=~\ind{A~<~-d}$.
Then, we have $C = C_0 + C_1 + C_2$.
Note from the theorem statement that $A = a_0^0 + a_1^0$ and $A = a_0^0 + a_1^0 - 2 \cdot \ringsz^0$ for the cases corresponding to rows $1$ and $8$ respectively from the table, while $A = a_0^0 + a_1^0 - \ringsz^0$ for the rest of cases.
Thus, it is easy to see that $c^0 = \fidiv{A}{d}$.
Also note that $-2 \cdot d + 2 \leq A \leq 2 \cdot d - 2$, implying that the range of $c^0$ is $\{ -2, -1, 0, 1 \}$.
Now we look at each value assumed by $c^0$ separately as follows:
\begin{itemize}
    \item $c^0 = -2$: In this case, we have $(A < -d)$, implying $C_0 = C_1 = C_2 = 1$, and $1 - C = -2$.
    \item $c^0 = -1$: In this case, we have $(-d \leq A < 0)$, implying $C_0 = C_1 = 1, C_2 = 0$ and $1 - C = -1$.
    \item $c^0 = 0$: In this case, we have $(0 \leq A < d)$, implying $C_0 = 1, C_1 = C_2 = 0$ and $1 - C = 0$.
    \item $c^0 = 1$: In this case, we have $(d \leq A)$, implying $C_0 = C_1 = C_2 = 0$ and $1 - C = 1$.
\end{itemize}
Thus, $c =_{\ringsz} \corr \cdot \ringsz^1 + (1 - C) - B =_{\ringsz} \frdiv{a}{d} - \frdiv{\share{a}{\ringsz}{0}}{d} - \frdiv{\share{a}{\ringsz}{1}}{d}$.
\end{proof}
 \begin{algorithm}
\caption{Integer ring division, $\protdivring{\ringsz,d}$:}
\label{algo:div-ring}
\begin{algorithmic}[1]

\Require For $b \in \zo$, $\party{b}$ holds $\share{a}{\ringsz}{b}$, where $a \in \bbZ_{\ringsz}$.
\Ensure For $b \in \zo$, $\party{b}$ learns $\share{z}{\ringsz}{b}$ s.t. $z = \frdiv{a}{d}$.

\vspace{0.2cm}

    \State For $b \in \zo$, let $a_b, a_b^0, a_b^1 \in \bbZ$ and $\ringsz^0, \ringsz^1, \ringsz' \in \bbZ$ be as defined in \theoremref{general-division}. Let $\eta = \lceil \log (n) \rceil, \delta = \lceil \log 6d \rceil$, and $\Delta = 2^{\delta}$.
    \State For $b \in \zo$, $\party{b}$ invokes $\fdreluring{\ringsz}$ with input $\share{a}{\ringsz}{b}$ to learn output $\share{\alpha}{B}{b}$. Party $\party{b}$ sets $\share{m}{B}{b} = \share{\alpha}{B}{b} \xor b$.\State For $b \in \zo$, $\party{b}$ sets $x_b = \ind{\share{a}{\ringsz}{b} \geq \ringsz'}$.
    \State \alice samples $\share{\corr}{\ringsz}{0} \getsr \bbZ_{\ringsz}$ and $\share{\corr}{\Delta}{0} \getsr \bbZ_{\Delta}$.
    \For{$j = \{00, 01, 10, 11\}$}
        \State \alice computes $t_{j} = (\share{m}{B}{0} \xor j_0 \xor x_0) \wedge (\share{m}{B}{0} \xor j_0 \xor j_1)$ s.t. $j = (j_0 \concat j_1)$.
        \If{$t_j \wedge \ind{x_0 = 0}$}  \label{gen-corr-case-one}
            \State \alice sets $s_{j} =_{\ringsz} -\share{\corr}{\ringsz}{0}-1$ and $r_{j} =_{\Delta} -\share{\corr}{\Delta}{0}-1$.
        \ElsIf{$t_j \wedge \ind{x_0 = 1}$} \label{gen-corr-case-two}
            \State \alice sets $s_{j} =_{\ringsz} -\share{\corr}{\ringsz}{0}+1$ and $r_{j} =_{\Delta} -\share{\corr}{\Delta}{0}+1$.
        \Else
            \State \alice sets $s_{j} =_{\ringsz} -\share{\corr}{\ringsz}{0}$ and $r_{j} =_{\Delta} -\share{\corr}{\Delta}{0}$.
        \EndIf
    \EndFor
    \State \alice \& \bob invoke an instance of $\kkot{4}{\etax+\delta}$ where \alice is the sender with inputs $\{s_j \concat r_j \}_{j}$ and \bob is the receiver with input $\share{m}{B}{1} \concat x_1$. \bob sets its output as $\share{\corr}{\ringsz}{1} \concat \share{\corr}{\Delta}{1}$. \label{div-corr}
    \State For $b \in \zo$, $\party{b}$ sets $\share{A}{\Delta}{b} =_{\Delta} a_b^0 - (x_b - \share{\corr}{\Delta}{b}) \cdot \ringsz^0$. \label{share-A}
    \State For $b \in \zo$, $\party{b}$ sets $\share{A_0}{\Delta}{b} =_{\Delta} \share{A}{\Delta}{b} - b \cdot d$, $\share{A_1}{\Delta}{b} = \share{A}{\Delta}{b}$, and $\share{A_2}{\Delta}{b} =_{\Delta} \share{A}{\Delta}{b} + b \cdot d$. \label{share-As}
    \For{$j = \{0, 1, 2\}$} \label{three-div-comps}
        \State For $b\in \zo$, $\party{b}$ invokes $\fdreluint{\delta}$ with input $\share{A_j}{\Delta}{b}$ to learn output $\share{\gamma_j}{B}{b}$. Party $\party{b}$ sets $\share{C'_j}{B}{b} = \share{\gamma_j}{B}{b} \xor b$. \label{div-drelu}
        \State  For $b \in \zo$, $\party{b}$ invokes an instance of $\fBtoA{\ringsz}$ with input $\share{C'_j}{B}{b}$ and learns $\share{C_j}{\ringsz}{b}$.  \label{div-B2A}
    \EndFor
    \State For $b \in \zo$, $\party{b}$ sets $\share{C}{\ringsz}{b} = \share{C_0}{\ringsz}{b} + \share{C_1}{\ringsz}{b} + \share{C_2}{\ringsz}{b}$. \label{div-C}
    \State For $b \in \zo$, $\party{b}$ sets $B_b = \fidiv{a_b^0 - x_b \cdot \ringsz^0}{d}$. \label{div-B}
    \State $\party{b}$ sets $\share{z}{\ringsz}{b} =_{\ringsz} \frdiv{\share{a}{\ringsz}{b}}{d} + \share{\corr}{\ringsz}{b} \cdot \ringsz^1 + b - \share{C}{\ringsz}{b} - B_b$, for $b \in \zo$.

\end{algorithmic}
\end{algorithm}
 
\section{Protocol for general division}
\label{app:division-protocol}

We describe our protocol for general division formally in \algoref{div-ring}.
As discussed in \sectionref{prot-general-division}, our protocol builds on \theoremref{general-division} and we compute the various sub-terms securely using our new protocols. Let $\delta= \lceil \log 6d \rceil$. We compute the shares of $\corr$ over both $\bbZ_\ringsz$ and $\bbZ_\Delta$ (Step~\ref{div-corr}). 
We write the term $C$  as $(\drelu(A- d) \xor 1) + (\drelu(A) \xor 1) + (\drelu(A+ d) \xor 1)$, which can be computed using three calls to $\fdreluint{\delta}$ (Step~\ref{div-drelu}) and $\fBtoA{\ringsz}$ (Step~\ref{div-B2A}) each.
\\

\noindent{\em Correctness and Security.} 
First, $m = \reconst{B}{\share{m}{B}{0}, \share{m}{B}{1}} =$ \\ $\reconst{B}{\share{\alpha}{B}{0}, \share{\alpha}{B}{1}} = \ind{a \geq n'}$.  Next, similar to \algoref{truncate-int}, $\reconst{\intx}{\share{\corr}{\intx}{0}, \share{\corr}{\intx}{1}} = \corr = \reconst{\Delta}{\share{\corr}{\Delta}{0}, \share{\corr}{\Delta}{1}}$, where $\corr$ is as defined in \theoremref{general-division}. Given the bounds on value of $A$ (as discussed above), it easy to see that Steps~\ref{share-A}\&\ref{share-As} compute arithmetic shares of $A$, and $A_0 = (A-d), A_1=A, A_2=(A+d)$, respectively. Now, invocation of $\fdreluint{\delta}$ on shares of $A_j$ (Step~\ref{div-drelu}) returns boolean shares of $\gamma = (1\xor \msbl(A_j))$ over $\delta$ bit integers, which is same as $1 \xor \ind{A_j < 0}$ over $\bbZ$. Hence, $C'_j = \reconst{B}{\share{C'_j}{B}{0}, \share{C'}{B}{1}} = \ind{A_j < 0}$. By correctness of $\fBtoA{\ringsz}$, step~\ref{div-C} computes arithmetic shares of $C$ as defined in \theoremref{general-division}. In step~\ref{div-B}, $B_0 + B_1 =_\ringsz B$ as defined. Hence, correctness holds and $\share{z}{\ringsz}{b}$ are shares of $\frdiv{a}{d}$. 

Given that $\share{\corr}{\ringsz}{0}$ and $\share{\corr}{\Delta}{0}$ are uniformly random, security of the protocol is easy to see in $(\kkot{4}{\etax+\delta}, \fdreluint{\delta},\fBtoA{\ringsz})$-hybrid. \\

\noindent{\em Communication complexity.}
$\protdivring{\ringsz,d}$ involves a single call to $\fdreluring{\ringsz}$ and $\kkot{4}{\etax+\delta}$, and three calls each to $\fdreluint{\delta}$ and $\fBtoA{\ringsz}$.
From \appendixref{relu-protocol}, we have the cost of $\fdreluring{\ringsz}$ as $\frac{3}{2}\secpar\etax+27\etax-\frac{\secpar}{2}-13$ bits.
$\kkot{4}{\etax+\delta}$ and $3 \times \fBtoA{\ringsz}$ cost $2\secpar + 4\cdot(\etax+\delta)$ and $3\secpar + 3\etax$ bits respectively.
Since the cost of $\fdreluint{\ell}$ is $\secpar \ell + 13\frac{1}{2}\ell - 3\secpar - 35\frac{1}{2}$ bits (see \appendixref{relu-protocol}), $3 \times \fdreluint{\delta}$ requires $3 \secpar \delta + 40\frac{1}{2}\delta - 9\secpar - 106\frac{1}{2}$ bits of communication.
Thus, the overall communication of $\protdivring{\ringsz,d}$ is $\frac{3}{2} \secpar \etax + 34 \etax + 3 \secpar \delta + 44\frac{1}{2} \delta - 4\frac{1}{2} \secpar - 119\frac{1}{2}$, which can be rewritten as $< (\frac{3}{2} \secpar + 34) \cdot (\etax + 2 \delta)$.
Concretely, we get the best communication for $\protdivring{\ringsz,49}$ ($\etax=32$) by setting $\blsize=7$ in all our millionaire invocations, which results in a total communication of $7796$ bits.

Note that for the case of $\ell$-bit integers, our division protocol would require a call to $\fdreluint{\ell}$ and $\kkot{4}{\ell+\delta}$, and three calls each to $\fdreluint{\delta}$ and $\fBtoA{\intx}$.
The cost of $\fdreluint{\ell}$ and $3\times\fdreluint{\delta}$ are as mentioned in the previous paragraph, and the cost of $\kkot{4}{\ell+\delta}$ and $\fBtoA{\intx}$ are $2\secpar + 4\cdot(\ell+\delta)$ and $3\secpar + 3\ell$ bits respectively.
Thus, the overall communication is $\secpar \ell + 3 \secpar \delta + 20\frac{1}{2} \ell + 44\frac{1}{2} \delta - 7 \secpar - 142$ bits, which can be rewritten as $< (\secpar + 21) \cdot ( \ell + 3 \delta )$.
By setting $\blsize=8$ in all our millionaire invocations, we get the best communication of $5570$ bits for $\protdivint{32,49}$.
\section{Input Encoding}
\label{app:input-encoding}

Neural network inference performs computations on floating-point numbers, whereas the secret-sharing techniques only work for integers in a ring $\bbZ_{\ringsz}$, for any $\ringsz \in \mathbb{N}$.\footnote{Note that this includes the case of $\ell$-bit integers when $\ringsz = 2^{\ell}$.}

To represent a floating-point number $x \in \mathbb{Q}$ in the ring $\bbZ_{\ringsz}$, we encode it as a fixed-point integer $a = \lfloor x \cdot 2^s \rfloor \bmod{\ringsz}$ with scale $s$.
Fixed-point arithmetic is performed on the encoded input values (in the secure domain) and the same scale $s$ is maintained for all the intermediate results.
The ring size $\ringsz$ and the scale $s$ are chosen such that the absolute value of any intermediate result does not exceed the bound $\lfloor \ringsz/2 \rfloor$ and there is no loss in accuracy (refer \appendixref{accuracy-summary}).

\section{Improvement to Gazelle's Algorithm} \label{app:opti-gazelle-algo}
Gazelle~\cite{gazelle} proposed two methods for computing convolutions, namely, the input rotations and the output rotations method.
The only difference between the two methods is the number of (homomorphic) rotations required\footnote{The number of homomorphic additions also differ, but they are relatively very cheap.}.
In this section, we describe an optimization to reduce the number of rotations required by the output rotations method.

Let $c_i$ and $c_o$ denote the number of input and output channels respectively, and $c_n$ denote the number of channels that can fit in a single ciphertext.
At a high level, the output rotations method works as follows:
after performing all the convolutions homomorphically, we have $c_i \cdot c_o/c_n$ intermediate ciphertexts that are to be accumulated to form tightly packed output ciphertexts.
Since most of these ciphertexts are misaligned after the convolution, they must be rotated in order to align and pack them.
The intermediate ciphertexts can be grouped into $c_o/c_n$ groups of $c_i$ ciphertexts each, such that the ciphertexts within each group are added (after alignment) to form a single ciphertext.
In~\cite{gazelle}, the ciphertexts within each group are rotated (aligned) individually, resulting in $\approx c_i \cdot \frac{c_o}{c_n}$ rotations.
We observe that these groups can be further divided into $c_n$ subgroups of $c_i/c_n$ ciphertexts each, such that ciphertexts within a subgroup are misaligned by the same offset.
Doing this has the advantage that the $c_i/c_n$ ciphertexts within each subgroup can first be added and then the resulting ciphertext can be aligned using a single rotation.
This brings down the number of rotations by a factor of $c_i/c_n$ to $\approx c_n \cdot \frac{c_o}{c_n}$.

With our optimization, the output rotations method is better than the input rotations method when $f^2 \cdot c_i > c_o$, where $f^2$ is the filter size, which is usually the case.
 \section{Complexity of our benchmarks}
\label{app:DNN-summary}

The complexity of the benchmarks we use in \sectionref{experiments} is summarized as follows:
\begin{itemize}
	\item \squeezenet: There are 26 convolution layers of maximum filter size $3\times3$ and up to 1000 output channels. The activations after linear layers are $\relu$s with size of up to 200,704 elements per layer. All $\relu$ layers combined have a size of 2,033,480. Additionally, there are 3 $\maxpool$ layers and an $\avgpool_{169}$ layer ($\avgpool$ with pool size 169).
	\item \resnet: There are 53 convolution layers of maximum filter size $7\times7$ and a peak output channel count of 2048. Convolution layers are followed by batch normalization and then $\relu$s. There are 49 $\relu$ layers totaling 9,006,592 $\relu$s, where the biggest one consists of 802,816 elements. Moreover, $\resnet$ also has $\maxpool$ layers and an $\avgpool_{49}$.
	\item \densenet: There are 121 convolution layers with maximum filter dimension of $7\times7$ and up to 1000 output channels. Similar to \resnet, between 2 convolution layers, there is batch normalization followed by $\relu$. The biggest $\relu$ layer in $\densenet$ has 802,816 elements and the combined size of all $\relu$ layers is 15,065,344. In addition, $\densenet$ consists of a $\maxpool$, an $\avgpool_{49}$ and 3 $\avgpool_{4}$ layers.
\end{itemize}
 \section{Garbled circuits vs our protocols for $\avgpool$}
\label{app:avgpool-numbers}

\begin{table}[t]
    \centering
    \begin{subtable}{0.5\textwidth}
        \centering
        \begin{tabular}{|c|c|c|c|c|c|c|}
            \hline
            \multirow{2}{*}{Benchmark} & \multicolumn{3}{c|}{Garbled Circuits} & \multicolumn{3}{c|}{Our Protocol} \\
                           & LAN   & WAN    & Comm   & LAN  & WAN  & Comm \\ \hline \hline
            \squeezenet    & 0.2  & 2.0   & 36.02   & 0.1 & 0.8 & 1.84 \\ \hline
            \resnet        & 0.4  & 3.9   & 96.97   & 0.1 & 0.8 & 2.35 \\ \hline
            \densenet      & 17.2 & 179.4 & 6017.94 & 0.5 & 3.5 & 158.83 \\ \hline
        \end{tabular}
        \caption{over $\bbZ_{2^{\ell}}$} \label{subtab:exp-div-ring}
    \end{subtable}\\\smallskip
    \begin{subtable}{0.5\textwidth}
        \centering
        \begin{tabular}{|c|c|c|c|c|c|c|}
            \hline
            \multirow{2}{*}{Benchmark} & \multicolumn{3}{c|}{Garbled Circuits} & \multicolumn{3}{c|}{Our Protocol} \\
                          & LAN   & WAN    & Comm   & LAN  & WAN  & Comm \\ \hline \hline
            \squeezenet   & 0.2  & 2.2   & 39.93   & 0.1 & 0.9 & 1.92 \\ \hline
            \resnet       & 0.4  & 4.2   & 106.22  & 0.1 & 1.0 & 3.82 \\ \hline
            \densenet     & 19.2 & 198.2 & 6707.94 & 0.6 & 4.4 & 214.94 \\ \hline
        \end{tabular}
        \caption{over $\bbZ_{\ringsz}$} \label{subtab:exp-div-field}
    \end{subtable}
    \vspace{-5pt}
    \caption{Performance comparison of Garbled Circuits with our protocols for computing $\avgpool$ layers. Runtimes are in seconds and communication numbers are in MiB.} \label{tab:exp-div}
\end{table}
 
In this section, we compare our protocols with garbled circuits for evaluating the $\avgpool$ layers of our benchmarks, and the corresponding performance numbers are given in \tableref{exp-div}.
On $\densenet$, where a total of $176,640$ divisions are performed, we have improvements over GC of more than $32 \times$ and $45 \times$ in the LAN and the WAN setting, respectively, for both our protocols.
However, on $\squeezenet$ and \resnet, the improvements are smaller ($2\times$ to $7\times$) because these DNNs only require $1000$ and $2048$ divisions, respectively, which are not enough for the costs in our protocols to amortize well.
On the other hand, the communication difference between our protocols and GC is huge for all three DNNs.
Specifically, we have an improvement of more than $19\times$, $27\times$, and $31\times$ on \squeezenet, \resnet, and \densenet respectively, for both our protocols.
 \balance
\section{Fixed-point accuracy of our benchmarks}
\label{app:accuracy-summary}

In this section, we show that the accuracy achieved by the fixed-point code matches the accuracy of the input TensorFlow  code. \tableref{accuracy-tab} summarizes the bitwidths, the scales, and the corresponding TensorFlow (TF) and fixed-point accuracy for each of our benchmarks. 	Since our truncation and division protocols lead to faithful implementation of fixed-point arithmetic, accuracy of secure inference is the same as the fixed-point accuracy.

\begin{table}[t]
\centering
\begin{tabular}{|c|c|c|c|c|c|c|}
\hline
    \multirow{2}{*}{Benchmark} & \multirow{2}{*}{Bitwidth} & \multirow{2}{*}{Scale} & TF & Fixed & TF & Fixed \\
    & & & Top 1 & Top 1 & Top 5 & Top 5 \\
\hline
\hline
$\squeezenet$  & 32 & 9  & 55.86 & 55.90 & 79.18 & 79.22 \\ \hline
$\resnet$      & 37 & 12 & 76.47 & 76.45 & 93.21 & 93.23 \\ \hline
$\densenet$    & 32 & 11 & 74.25 & 74.35 & 91.88 & 91.90 \\ \hline
\end{tabular}
\caption{Summary of the accuracy achieved by fixed-point code vs input TensorFlow (TF) code.}
\label{tab:accuracy-tab}
\end{table}

\end{document}